\documentclass{article}

\usepackage[dvips]{epsfig}
\usepackage{graphicx}
\usepackage{float}

\usepackage{ulem}

\usepackage{amsmath}
\usepackage{amssymb}
\usepackage{amsbsy}
\usepackage{bbm}
\usepackage{eucal}

\usepackage{enumitem}

\usepackage{hyperref}

\usepackage{dsfont}
\usepackage{palatino}
\usepackage{xcolor}

\usepackage{multirow}

\usepackage{amsthm}

\newtheorem{theorem}{Theorem}

\newtheorem{assumption}[theorem]{Assumption}
\newtheorem{lemma}[theorem]{Lemma}

\newtheorem{corollary}[theorem]{Corollary}
\newtheorem{definition}[theorem]{Definition}
\newtheorem{proposition}[theorem]{Proposition}

\usepackage{amscd}

\usepackage{tensor}

\newcommand{\norm}[1]{ {|\!|  #1 |\!|}}

\newcommand{\opnorm}[1]{{\vert\kern-0.25ex\vert\kern-0.25ex\vert #1 
  \vert\kern-0.25ex\vert\kern-0.25ex\vert}}

\newcommand{\1}{\mathbbm{1}}

\newcommand{\MINL}{\mathcal{U}} 

\allowdisplaybreaks

\title
{The two square root laws of market impact \\
and the role of sophisticated market participants}

\author
{Bruno Durin\footnote{bruno.durin@gmail.com} \and
Mathieu Rosenbaum\footnote{CMAP, mathieu.rosenbaum@polytechnique.edu} \and Gr\'egoire Szymanski\footnote{CMAP, gregoire.szymanski@polytechnique.edu}}

\begin{document}

\date{\today}

\maketitle

%

\begin{abstract} 
The goal of this paper is to disentangle the roles of volume and of participation rate in the price response of the market to a sequence of transactions. To do so, we are inspired the methodology introduced in \cite{jaisson2015market, jusselin2020noarbitrage} where price dynamics are derived from order flow dynamics using no arbitrage assumptions. We extend this approach by taking into account a sophisticated market participant having superior abilities to analyse market dynamics. 
Our results lead to the recovery of two square root laws: (i) For a given participation rate, during the execution of a metaorder, the market impact evolves in a square root manner with respect to the cumulated traded volume. (ii) For a given executed volume $Q$, the market impact is proportional to $\sqrt{\gamma}$, where $\gamma$ denotes the participation rate, for $\gamma$ large enough. Smaller participation rates induce a more linear dependence of the market impact in the participation rate.
\end{abstract}



\section{Introduction}

Market impact refers to the fact that buy orders push on average the price up and sell orders push it down. In particular, large orders, called metaorders, split over time in smaller orders, induce a liquidity imbalance resulting in an adversarial and mechanistic movement in prices. Market impact stands out as a prominent transaction cost associated with the execution of metaorders \cite{almgren2005direct,freyre2004review,hey2023cost,robert2012measuring}. The measurement and understanding of market impact have thus emerged as a central theme in quantitative finance, see \cite{webster2023handbook} for a review.
\\

Market impact is intrinsically hard to study because of its noisy nature. During the execution of a given metaorder, many other orders are likely being traded simultaneously. Although it requires a very careful statistical treatment, averaging over many metaorders eliminates part of this noise so that we can better identify universal properties of market impact, see for instance \cite{almgren2005direct, bacry2015market, bershova2013non, bucci2019crossover}. When trading a metaorder, these empirical studies show that price mechanically follows the metaorder, exhibiting a concave shape peaking at the end of the metaorder, followed by a convex relaxation, see \cite{bershova2013non, gatheral2010no, moro2009market}. \\

A crucial issue in the study of market dynamics is the dependence of market impact on the volume of metaorders. This question is extensively studied in the literature and various authors have identified a square root dependence in the volume \cite{almgren2005direct, hopman2003essays,kyle2023large, lillo2003master,  moro2009market}. Let us consider a metaorder executed on $[0,T]$. The market impact at time $t$ depends on the quantity $Q_t$ that is executed before time $t$ (so $Q_T$ is the total volume of the metaorder.) The square root law states that
\begin{equation}
\label{eq:pure_square_rootmarket_impact_volume}
MI(t, Q_t) \approx c \sigma \Big( \frac{Q_t}{V} \Big)^{1/2},
\end{equation}
for all $t\leq T$ where $c$ is a constant of order $1$, $\sigma$ is the daily volatility and $V$ the usual daily volume traded on the asset in normal conditions. This relationship is at first sight counter-intuitive: one could first think of a linear impact in volume $Q$ since two consecutive metaorders of volume $Q/2$ would then create the same peak market impact, as in \cite{kyle1985continuous}. Further studies identify that this linearity can actually hold in some sense (see below) but only for small orders, while large orders seem to follow the square root law of Equation \eqref{eq:pure_square_rootmarket_impact_volume}. We can summarize this behaviour by saying that the market impact is concave in the volume. In \cite{benzaquen2018market}, the authors proposed a more accurate approximation of the market impact of the form
\begin{equation}
\label{eq:market_impact_volume}
MI(t, Q_t) \approx c \sigma \Big( \frac{Q_t}{V} \Big)^{1/2} \mathcal{F}\Big(\frac{Q_t}{Vt} \Big),
\end{equation}
where $\mathcal{F}$ is monotonic and satisfies $\mathcal{F}(x) \approx \sqrt{x}$ when $x\to0$ and $\mathcal{F}(x) \to a$ when $x\to\infty$ for some $a > 0$. In \cite{benzaquen2018market}, a theoretical explanation for this observed relationship is also provided. The idea is rooted in the dynamic theory of market liquidity, which is also discussed in \cite{mastromatteo2014agent, toth2011anomalous}. Earlier studies often assume a constant liquidity profile across prices, implying a similar supply and demand at any given price. However, it has become clear that the liquidity profile should exhibit a 'V'-shaped pattern, diminishing in the vicinity of the current price while linearly increasing as one moves away from it. This peculiar shape arises because a significant portion of liquidity remains latent and concealed from traders, not directly visible in the limit order book. Instead, orders are placed only when the current market price approaches the level where this concealed liquidity becomes accessible. In \cite{benzaquen2018market}, the liquidity is modelled using the hydrodynamic limit of the latent order book as described in \cite{donier2015fully}. This approach leads to a partial differential equation, the solution of which provides the relationship expressed in \eqref{eq:market_impact_volume}.
\\

An important parameter to understand market impact is the participation rate $\gamma$, defined as the proportion of the market volume associated to the metaorder. Considering that the execution rate of the order is constant, we have $Q_t = \gamma V t$. Using this parametrisation, \eqref{eq:market_impact_volume} becomes
\begin{equation}
\label{eq:market_impact_gamma}
MI(t, Q_t) = c \sigma \Big( \frac{Q_t}{V} \Big)^{1/2} \mathcal{F}(\gamma).
\end{equation}
Interestingly, this equation encompasses two square root laws for market impact. The first states that for a fixed participation rate $\gamma$, during the execution of a metaorder, the market impact evolves in a square root manner with respect to the cumulated traded volume. The second is that at peak impact, the market impact of market orders of similar maturity is proportional to $\sqrt{\gamma}\mathcal{F}(\gamma)$. Since $\mathcal{F}(x) \approx \sqrt{x}$ when $x\to0$ and $\mathcal{F}(x) \to a$ when $x\to\infty$, the market impact is proportional to $\sqrt{\gamma}$ when $\gamma$ is large, and is linear in $\gamma$ when $\gamma$ is small.
\\

The goal of this paper is to understand the financial mechanisms leading to these relationships. We show that no-arbitrage constraints together with the diversity of market participants generate these subtle responses of market prices to the order flow. To do so, we use Hawkes processes which provide a natural framework to model arrival times of transactions
\cite{filimonov2012quantifying, filimonov2015apparent, hardiman2013critical, jaisson2015limit, jaisson2016rough}. Here, we study market impact in the spirit of \cite{jaisson2015market,jusselin2020noarbitrage}.
In this framework, price dynamics are obtained as an anticipation of the market towards the future order flow.
More precisely, in \cite{jaisson2015market}, the order flow is modelled by two independent Hawkes processes $N^a$ and $N^b$ whose jump times are the arrival times of buy and sell market orders on the market. We denote by $\varphi$ the self-exciting kernel of $N^a$ and $N^b$ and we define $\psi = \sum_{k\geq 1} \varphi^{*k}$ where $\varphi^{*k}$ stands for the $k$ fold convolution of $\varphi$. It is shown in the same work that for the price to satisfy a no-statistical arbitrage condition together with a martingale behaviour, it has to evolve as follows
\begin{equation*}
P_t 
= 
P_0 +
\lim\limits_{s\to\infty} \kappa \, \mathbb{E}[N^{a}_s - N^{b}_s \, | \, \mathcal{F}_s ]
=
P_0 + \kappa 
\int_0^t 
\xi(t-s) \, d (N^{a}_s - N^{b}_s)
\end{equation*}
where $\xi(t) = 1 + ( 1 + \int_0^\infty \psi(s) \, ds ) \int_t^\infty \varphi(s) \, ds$. In particular, note that $\xi(t) \to 1$ as $t \to \infty$ and therefore the permanent impact of a single order is given by $\kappa$, see \cite{jaisson2015market} for more details. Let us consider now a metaorder executed in this market modelled by a point process $N^o$. A priori, the market cannot make any difference between this metaorder and the orders originating from the previous order flow and therefore digests it as any other order. In this case, the price becomes 
\begin{equation}
\label{eq:price:impact:intro}
P_t 
= 
P_0 + \kappa 
\int_0^t 
\xi(t-s) \, d (N^{a}_s - N^{b}_s + N^o_s).
\end{equation}
The relaxation of the market impact after its peak is due to the presence of $\xi(t) - 1  = ( 1 + \int_0^\infty \psi(s) \, ds ) \int_t^\infty \varphi(s) \, ds$ in \eqref{eq:price:impact:intro} and its characteristic square root shape is linked to the interraction between the criticality of $\phi$ (through the norm $\norm{\psi}_{L^1} = (1 - \norm{\varphi}_{L^1})^{-1} \norm{\varphi}_{L^1} $) and long memory of the Hawkes processes, see \cite{jusselin2020noarbitrage} for more details. However, this fails to reproduce the dependence in the participation of the market impact described in Equation \eqref{eq:market_impact_gamma}.
\\

In practice, the orders coming from a large metaorder induce a price drift and a liquidity imbalance which can be detected by some sophisticated market participant. In particular, this participant knows that the long time drift of the price due to this metaorder should be equal to $\kappa$ times the actual size of this metaorder. In that case, they can compute a fair price $EP_t$ as follows
\begin{equation*}
EP_t 
= 
P_0 +
\lim\limits_{s\to\infty} \kappa \, \mathbb{E}[N^{a}_s - N^{b}_s + N^{o}_s \, | \, \mathcal{F}_s ]
=
P_0 + \kappa 
\int_0^t 
\xi(t-s) \, d (N^{a}_s - N^{b}_s) + \kappa N^{o}_t.
\end{equation*} 
This participant can then exploit the market's over-response to the metaorder taking advantage of the difference between $P_t$ and $EP_t$. To model this, we assume that they trade according to the magnitude of this signal sending sell order with 
%
%
intensity $\MINL(P_t - EP_t)$ for some increasing function $\MINL$. Note that such type of dynamics are related to quadratic Hawkes processes introduced in \cite{blanc2017quadratic} and rough volatility \cite{dandapani2021quadratic, gatheral2018volatility,gatheral2020quadratic}. 
\\

This approach also bears a strong connection with the dynamic theory of market liquidity, where the square root law emerges through modifications to the liquidity profile and the specification of its V-shape dynamics \cite{benzaquen2018market, bucci2019crossover, donier2015fully}. In our context, the liquidity is implicitly altered due to the presence of a sophisticated trader. Specifically, their reaction to the metaorder mitigates the usual market response and pushes the price to its original level by introducing additional orders. It is important to distinguish this mechanism from the one described in \cite{benzaquen2018market}, where the additional liquidity is posted on the same side as the metaorder. In contrast, the sophisticated trader in our context acts by executing market orders in the opposite direction. This apparent contradiction can be resolved by noting that the latent liquidity in \cite{benzaquen2018market} can be viewed as latent limit orders, while the sophisticated trader uses here market orders. \\

Taking into account the sophisticated trader introduced previously, we retrieve the two square root laws of market impact in Section \ref{sec:macroscopic:limit}. More precisely, we show that during the execution given metaorder, the market impact evolves following a square root pattern in relation to the cumulative traded volume, exhibiting then a convex relaxation. Moreover, we prove that under mild hypothesis on $\MINL$, the market impact is concave in the participation rate for fixed duration $t$. The strength of our approach resides in its bottom-up methodology, which effectively captures the desired stylized facts through a constructive representation of market microstructure dynamics. More precisely, we consider the scaling limit of the Hawkes order flow model in the presence of a sophisticated market participant reacting to the metaorder in the spirit of \cite{jaisson2016rough, jusselin2020noarbitrage}. Our result rests upon mild assumptions in agreement with financial data. Specifically, we assume that the kernel $\varphi$ han an $L^1$ norm close to $1$ \cite{filimonov2012quantifying, filimonov2015apparent}  and that it has a power law decay \cite{hardiman2013critical}. Additionally, we impose the only constraints on the parameters ensuring the scaling limit is non-degenerate. Our methodology also has practical use and can be applied by practitioners. Following a suitable calibration of the parameters $(\varphi, \MINL)$, one can leverage Monte Carlo simulations to assess market impact profiles for various execution strategies, even in scenarios involving substantial participation rates. We also identify a closed form formula for the scaling limit of large metaorders that allows for a precise analysis of the shape of market impact. \\

This paper is organised as follows. In Section \ref{sec:marketimpactsmart}, we detail the construction of the Hawkes market impact model with a sophisticated trader. Section \ref{sec:macroscopic:scaling} is dedicated to the computation of the scaling limit of the market impact. We further investigate this limit in Section \ref{sec:macroscopic:limit} where we establish the two square root laws of market impact and explore various other stylized facts. The proofs are gathered in Appendices \ref{sec:proof:whole_model}, \ref{sec:proof:market_impact} and \ref{sec:proof:macroscopic:limit}.\\

\section{Market impact in the presence of a sophisticated participant}
\label{sec:marketimpactsmart}

\subsection{Hawkes order flow in absence of metaorders}

We assume that the buy and sell market order flow follows a Hawkes process as in \cite{jaisson2015market, jusselin2020noarbitrage}. The arrival times of buy and sell market orders are the jump times of two independent Hawkes processes, denoted by $N^a$ (buy) and $N^b$ (sell). Both processes have the same baseline intensity $\mu \geq 0$ and self-exciting (non-negative) kernel $\varphi$ so that their intensity is given by
\begin{equation*}
\lambda^{a}_t = \mu + \int_0^{t-} \varphi(t-s) \, dN^{a}_s
\;\;\text{ and }\;\;
\lambda^{b}_t = \mu + \int_0^{t-} \varphi(t-s) \, dN^{b}_s.
\end{equation*}
To simplify our analysis, we assume that all market orders have the same volume, and we normalize this volume to one. In this setting, it is shown in \cite{jaisson2015market} that under no statistical arbitrage assumptions and a martingale condition, the market price, denoted by $P_t$, satisfies
\begin{equation}
\label{eq:market_price}
P_t 
= 
P_0 +
\lim\limits_{s\to\infty} \kappa \, \mathbb{E}[N^{a}_s - N^{b}_s \, | \, \mathcal{F}_t ]
=
P_0 + \kappa 
\int_0^t 
\xi(t-s) \, d (N^{a}_s - N^{b}_s)
\end{equation}
where $\kappa > 0$ is a constant, $\mathcal{F}_t$ represents the information available up to time $t$, and the function $\xi$ is given by
\begin{equation}
\label{eq:def:xi}
\xi(t) = 1 + \Big( 1 + \int_0^\infty \psi(s) \, ds \Big) \int_t^\infty \varphi(s) \, ds
\;\;\text{ and } 
\;\;\psi = \sum_{k=1}^\infty \varphi^{*k}
\end{equation}
where $\varphi^{*k}$ stands for the $k$ fold convolution of $\varphi$, \textit{i.e.} $\varphi^{*1} = \varphi$ and $\varphi^{*(k+1)} = \varphi^{*k}*\varphi$ where $f * g (x) = \int_\mathbb{R} f(y) g(x-y)\, dy$ denotes the convolution between $f$ and $g$. From \eqref{eq:market_price}, we see that the impact on the price at time $t$ of a single order happening at time $t_0 \leq t$  is $\kappa \xi(t-t_0)$. Since $\xi(t) \to 1$ as $t \to \infty$, $\kappa$ can also be seen as the permanent impact of a single order. The pre-factor $\kappa$ can also be linked to the total volatility of the market. More precisely, in \eqref{eq:market_price}, the realised volatility on $[0,T]$ would be $\kappa^2 (N^{a}_T + N^{b}_T)$. However, exogenous non-trading related price changes are ignored in this model as they play no role in our impact study. To take them into account in a statistical perspective, one could for instance introduce an additional Brownian motion $B$ independent of $N^{a}$ and $N^{b}$ and set
\begin{equation*}
P_t 
=
P_0 + \kappa 
\int_0^t 
\xi(t-s) \, d (N^{a}_s - N^{b}_s) + \widetilde\sigma B_t.
\end{equation*}
In this case, the realised volatility on $[0,T]$ would be $\kappa^2 (N^{a}_T + N^{b}_T) +  \widetilde\sigma^2 T$. We can estimate $\kappa^2$ and $\widetilde\sigma^2$ from this formula by regression. In a highly endogenous market, it is reasonable to expect that $\kappa \gg \widetilde\sigma$. In that case we would have $\kappa^2 (N^{a}_T + N^{b}_T) +  \widetilde\sigma^2 T \approx \kappa^2 (N^{a}_T + N^{b}_T)$, effectively leading to the prefactor $c\sigma$ in \eqref{eq:market_impact_gamma}. We furthermore remark that $\kappa \xi(0)$ essentially corresponds to the ratio between volatility and square root of volume per unit of time, which is the typical prefactor in the he impact law \eqref{eq:market_impact_gamma}.\\

Note also that this model focuses on metarders executed \textit{via} market orders. In practice, limit orders can also be used. To take into account the impact of limit orders, one has to compare the market dynamic in the presence of these limit orders and without them. Therefore, in that case, market impact is generated as long as the limit orders are posted in the order book, as they have an influence on the incoming order flow by other market participants and on the dynamics of price changes, see \cite{huang2015simulating}. Properly taking into account these effects requires making $\kappa$ dependent on the posted quantities and modelling the feedback between displayed quantities in the order book and the incoming order flow. Investigating this setting is left for future research.

\subsection{Execution of a metaorder and sophisticated trader's reaction}

We now consider a metaorder on the ask side. We model the arrival times of this metaorder as the jump times of a point process $N^o$ independent of $N^{a}$ and $N^{b}$. In particular, this implies that the jump times of $N^o$ are almost surely all different from the jump times of $N^{a}$ and $N^{b}$. In the spirit of \cite{jaisson2015market,jusselin2020noarbitrage}, we assume that $N^o$ is a nonhomogeneous Poisson process with intensity $\nu$. Since trades are anonymous, the market digests the metaorder in the same way as the original order flow. Therefore, the market price dynamic, as expressed in Equation \eqref{eq:market_price}, is modified to account for the presence of the metaorder $N^o$, resulting in 
\begin{equation}
\label{eq:market_price:metaorder}
P_t 
=
P_0 + \kappa 
\int_0^t 
\xi(t-s) \, d (N^{a}_s + N^{o}_s - N^{b}_s).
\end{equation}
However, it is crucial to note that in fact $N^{a}$ and $N^o$ should not have the same impact on the market since $N^{a}$ is self-exciting, while $N^o$ is not. From a market microstructure viewpoint, the metaorder creates some market imbalance that could be identified using various high-frequency trading tools and techniques, such as detecting pattern variations in the auto-correlation function of market order signs or leveraging clustering algorithms with machine learning to approximate the fair price, see \cite{toth2010segmentation, vaglica2008scaling} for more details. Therefore it seems reasonable that some sophisticated agents could at least partially disentangle $N^a$ and $N^o$ and compute the efficient market price $EP_t$, defined as
\begin{equation}
\label{eq:efficient_price:metaorder}
EP_t 
= 
P_0 +
\lim\limits_{s\to\infty} \kappa \, \mathbb{E}[N^{a}_s - N^{b}_s \, | \, \mathcal{F}_t ] + \kappa N^{o}_t
=
P_0 + \kappa 
\int_0^t 
\xi(t-s) \, d (N^{a}_s  - N^{b}_s) + \kappa N^{o}_t.
\end{equation}
Note from \eqref{eq:def:xi} that $\xi(t) \geq 1$ for any $t$, so that we have $EP_t \leq P_t$. Moreover, \eqref{eq:efficient_price:metaorder} can be viewed as a signal that a sophisticated trader uses to increase their profits. For small orders, $EP_t$ is close to $P_t$ and the signal should produce little to no reaction from the sophisticated trader. In fact, it is mostly important to compute efficiently $EP_t$ when it is far from $P_t$. This is the case for large metaorders, although even for these large metaorders, we still have $P_t \approx EP_t$ for small $t$ as long as the volume traded is small. Therefore, the precise detection of the metaorder, and in particular of its start, is not necessary. \\

In the following, we therefore consider a sophisticated trader succeeding in efficiently compute \eqref{eq:efficient_price:metaorder}. They may exploit the difference $EP_t \leq P_t$ by selling some assets. We model the arrival times of their reaction as the jump times of a point process $N^m$ whose intensity is given by 
\begin{equation*}
\lambda^{m}_t = \MINL \big( P_{t-} - EP_{t-} \big) 
\end{equation*}
for some function $\MINL: \mathbb{R} \to \mathbb{R}$. A natural choice for $\MINL$ can be easily obtained from heuristics. Suppose that the sophisticated trader observes a signal $\alpha$ and believes that the price impact follows the square root law. Then if he invests $x$ according to this signal $\alpha$, he expects to earn $\alpha x$ from this investment, but the price impact being $c \sqrt{x} x$ for some $c>0$, his effective PnL is $\alpha x - c \sqrt{x} x$. Optimizing this expression yields $x = \widetilde{c} \alpha^2$ for some $\widetilde{c} > 0$. Therefore, a natural choice would be $\MINL(x)$ proportional to $x^2$. However, in the following, we do not require that $\MINL(x)$ is proportional to $x^2$, see \cite{webster2023handbook} for discussion. Our results can be proved in a more general setting encompassing various trading strategies. More precisely, we only suppose that $\MINL$ is increasing (since a larger difference leads to a larger sell incentive for the sophisticated trader) and such that $\MINL(x) = 0$ for $x \leq 0$. Note that it would also be natural to assume that $\MINL$ exhibits convex properties. The idea behind this relies on the fact that when the difference reaches higher levels, the incentive should increase at an accelerated rate to reflect the necessity for the sophisticated participant to seize the priority by responding promptly. This assumption is not necessary however to retrieve our first scaling limit of the market impact and only becomes relevant when studying the dependence of the limit in the participation rate in Section \ref{sec:macroscopic:limit}.
\\ 

The sophisticated trader's actions also impact the market and alter both the market price dynamic and the fair price. Specifically, the market price dynamic, accounting for the effect of $N^m$ in an indistinguishable manner from $N^b$ and is given by  
\begin{equation}
\label{eq:market_price:full}
P_t 
=
P_0 + \kappa 
\int_0^t 
\xi(t-s) \, d (N^{a}_s + N^{o}_s - N^{b}_s - N^{m}_s).
\end{equation}
On the other hand, the sophisticated trader is aware of the non-excitating nature of trades from $N^{m}_s$, leading to the modification of the fair price
\begin{equation}
\label{eq:efficient_price:full}
EP_t 
=
P_0 + \kappa 
\int_0^t 
\xi(t-s) \, d (N^{a}_s  - N^{b}_s) + \kappa(N^{o}_t - N^{m}_t).
\end{equation}
Combining Equations \eqref{eq:market_price:full} and \eqref{eq:efficient_price:full}, we get
\begin{equation}
\label{eq:signal}
P_{t} - EP_{t} = 
\kappa 
\int_0^t 
\zeta(t-s) \, d (N^{o}_s - N^{m}_s)
\;
\text{ where }\;
\zeta(t) = \Big( 1 + \int_0^\infty \psi(s) \, ds \Big) \int_t^\infty \varphi(s) \, ds.
\end{equation}
In the above, the situation is clearly asymmetric, with arbitrageurs only selling stocks. This asymmetry arises because $EP_t \leq P_t$ always holds, even in the presence of this sophisticated participant. This is due to the fact that, almost surely, $N^o_t \geq N^m_t$ for any $t \geq 0$, and the function $\zeta$ is decreasing. We refer to Lemma~\ref{lemma:comparison} for a precise statement. Importantly, we do not impose any explicit risk constraints or inventory constraints on this trader. In practice, over the long term, we anticipate the occurrence of numerous metaorders, alternating between buy and sell. Therefore the signal oscillates and won't always have the same sign, ultimately leading to risk levels stabilizing and reaching equilibrium. 

\subsection{Rigorous definition of the impact model}

With these developments, we are now ready to formalize our setting in the subsequent definition.
\begin{definition}
\label{def:hawkesmarketimpact}
The \textbf{Hawkes market impact model with a sophisticated trader} with parameters \linebreak$(\mu, \varphi, \nu, \MINL)$ is the family of $(\mathcal{F}_t)_t$-adapted counting processes $(N^a, N^b, N^o, N^m)$ defined on a given probability space $(\Omega, \mathcal{F}, (\mathcal{F}_t)_t, \mathbb{P})$ such that
\begin{itemize}
\item The jump times of $(N^a, N^b, N^o, N^m)$ are almost surely mutually disjoint
\item $N^a$ and $N^b$ are independent Hawkes processes with baseline $\mu$ and self-exciting kernel $\varphi$. We write $\lambda^a$ and $\lambda^b$ their intensity.
\item $N^o$ is an inhomogeneous Poisson process with intensity $\lambda^o_t = \nu(t)$ where $\nu$ is a (deterministic) non-negative integrable function.
\item The compensator $\Lambda^m$ of $N^m$ has the form $\Lambda^m_t = \int_0^t  \lambda^m_s \, ds$ where $\lambda^m_t = \MINL \big( P_{t-} - EP_{t-} \big)$ and 
\begin{equation*}
\begin{cases}
P_t 
=
P_0 + \kappa 
\int_0^t 
\xi(t-s) \, d (N^{a}_s + N^{o}_s - N^{b}_s - N^{m}_s),
\\
EP_t 
=
P_0 + \kappa 
\int_0^t 
\xi(t-s) \, d (N^{a}_s  - N^{b}_s) + \kappa(N^{o}_t - N^{m}_t).
\end{cases}
\end{equation*}
\end{itemize}
\end{definition}

It is \textit{a priori} unclear whether a Hawkes market impact model with a sophisticated trader exists. We already know that we can build the processes $(N^a, N^b, N^o)$ provided $(\Omega, \mathcal{F}, (\mathcal{F}_t)_t, \mathbb{P})$ is rich enough, see for instance Theorem 6 in \cite{delattre2016hawkes}. Moreover, the distribution of $(N^a, N^b, N^o)$ is uniquely determined by Definition \ref{def:hawkesmarketimpact}. The following result shows that we can also build $N^m$ and its law is uniquely determined.

\begin{theorem}[Characterisation of the Hawkes market impact model with a sophisticated trader]
\label{thm:whole_model}
\leavevmode
\begin{enumerate}[label={(\roman*)}]
\item \label{thm:whole_model:ii} Suppose that $(\pi^a, \pi^b, \pi^o, \pi^m)$ are independent Poisson point measures on $[0,\infty) \times [0, \infty)$ defined on a filtered probability space $(\Omega, \mathcal{F}, (\mathcal{F}_t)_t, \mathbb{P})$
and adapted to the filtration $(\mathcal{F}_t)_t$. Then there exists a pathwise unique family $(N^a, N^b, N^o, N^m)$ satisfying 
\begin{equation}
\label{eq:def:Npoint}
N^x_t = \int_0^t \int_0^\infty \1_{z \leq \lambda^x_s} \, \pi^x(ds\,dz)
\end{equation}
for all $x \in \{a,b,o,m\}$. Moreover, this family satisfies Definition \ref{def:hawkesmarketimpact}.

\item \label{thm:whole_model:i} Suppose that $(N^a, N^b, N^o, N^m)$ is a Hawkes market impact model with a sophisticated trader as in Definition \ref{def:hawkesmarketimpact}. Then, up to a probability space enlargement, there exists independent Poisson point measures $(\pi^a, \pi^b, \pi^o, \pi^m)$ on $[0,\infty) \times [0, \infty)$ satisfying \eqref{eq:def:Npoint} for all $x \in \{a,b,o,m\}$.
\end{enumerate}
\end{theorem}

In Theorem \ref{thm:whole_model}, Point \ref{thm:whole_model:ii} establishes the existence of the model, while \ref{thm:whole_model:i} ensures that the distribution of $(N^a, N^b, N^o, N^m)$ is uniquely determined. Although the significance of \ref{thm:whole_model:i} may not be immediately apparent at this stage, we will explicitly use the construction presented in Equation \eqref{eq:def:Npoint} in the proofs, see Section \ref{sec:step:widetilde:lambda:oT}. Note also that without the independence condition on the measures $(\pi^a, \pi^b, \pi^o, \pi^m)$, point \ref{thm:whole_model:i} would be an easier implication from \ref{thm:whole_model:ii}.\\

Combining Equations \eqref{eq:market_price:full} and \eqref{eq:efficient_price:full}, we are now ready to define the pathwise market impact.
\begin{definition}
\label{def:marketimpact}
The pathwise \textbf{market impact} $MI_t$ at time $t$ of the metaorder is given by
\begin{equation*}
MI_t = \kappa 
\int_0^t 
\xi(t-s) \, d (N^{o}_s - N^{m}_s).
\end{equation*}
\end{definition}
Although we are mostly interested by the average market impact $\mathbb{E}[MI_t]$, the pathwise market impact is more convenient to study the scaling limit in the following sections. This is due to the non-linearity introduced by the sophisticated trader. This notation also differs from \eqref{eq:market_impact_volume} as we drop the dependence in $Q_t$. This comes from the fact that in our model, the volume treated at time $t$ is $Q_t = N^o_t \approx \int_0^t \nu(s)\,ds$.

\section{Scaling limit of the market impact}
\label{sec:macroscopic:scaling}

Let $T > 0$ be the final horizon time of the metaorder. We now rescale the market impact function as $T \to \infty$ in the context of the Hawkes market impact model with a sophisticated trader defined in Definition \ref{def:hawkesmarketimpact}.

More precisely, suppose that $(N^{a,T}, N^{b,T}, N^{o,T}, N^{m,T})$ is a Hawkes market impact model with a sophisticated trader with parameters $(\mu^{T}, \varphi^{T}, \nu^{T}, \MINL^{T})$. We write $P^T$ and $EP^T$ for the price and efficient price defined by 
\begin{equation}
\label{eq:def:prices}
\begin{cases}
P_t^{T}
=
P_0^{T} + \kappa 
\int_0^t 
\xi^T(t-s) \, d (N^{a,T}_s + N^{o,T}_s - N^{b,T}_s - N^{m,T}_s),
\\
EP_t ^T
=
P_0^{T} + \kappa 
\int_0^t 
\xi^T(t-s) \, d (N^{a,T}_s  - N^{b,T}_s) + \kappa (N^{o,T}_t - N^{m,T}_t)
\end{cases}
\end{equation}
where 
\begin{equation}
\label{eq:def:xit}
\xi^T(t) = 1 + \Big( 1 + \int_0^\infty \psi^T(s) \, ds \Big) \int_t^\infty \varphi^T(s) \, ds = 1 + \zeta^T(t)
\;\;\text{ and } 
\;\;\psi^T = \sum_{k=1}^\infty (\varphi^T)^{*k}.
\end{equation}
For simplicity, we rescale the price by taking $P_0^{T} = 0$. Additionally, we define the market impac $MI^T$, as in Definition \ref{def:marketimpact}. Before stating our main result, we provide a concise set of assumptions regarding the scaling of the parameters $(\mu^{T}, \varphi^{T}, \nu^{T}, \MINL^{T})$.\\

Our first assumption is on the parameters of the Hawkes order flow, which drives the price in the absence of a metaorder. The scaling limit of this model was previously studied in \cite{jaisson2016rough}, and we adopt the same framework here.

\begin{assumption}
\label{assumption:a}
There exists a function $\varphi$ such that $\varphi^T = a^T \varphi$ for some sequence $a^T \to 1$ and $\norm{\varphi}_{L^1}  = 1$. Moreover, there exists $0 < \alpha < 1$ such that the limits
\begin{equation*}
K = \lim\limits_{t \to \infty} t^\alpha \int_t^\infty \varphi(s) ds, 
\;\;
\lambda = \lim\limits_{T \to \infty} (1-\alpha) K^{-1} T^\alpha (1-a^T)
, 
\;\;
\text{ and }
\mu^* = \lim\limits_{T \to \infty} T^{1-\alpha} \mu^T 
\end{equation*}
are finite.
\end{assumption}

Remark that unlike \cite{jaisson2016rough}, the $a^T \to 1$ does not seem necessary here at first sight. However,  $a^T (1-a^T)^{-1}$ is the natural timeframe to study the autocorrelation of the Hawkes process $N^T$. Since we consider orders on a time interval $[0,T]$ with $T\to\infty$, it is natural to assume that $a^T (1-a^T)^{-1} \to \infty$ and therefore $a^T \to 1$.\\

Note that under Assumption \ref{assumption:a}, we have
\begin{equation*}
\int_0^\infty \psi^T(s) \, ds = (1-a^T)^{-1} a^T
\;\;
\text{ and }
\;\;
\zeta^T(t) =  (1-a^T)^{-1} a^T \int_t^\infty \varphi(s) \, ds
\end{equation*}
and the long-term average intensity of the Hawkes processes $N^{a, T}$ and $N^{b, T}$ is $\beta^T = (1-a^T)^{-1}\mu^T$. Therefore the average number of trades from $N^{a,T}$ and $N^{b,T}$ on $[0,T]$ scales as $T \beta^T$. Remark that under Assumption \ref{assumption:a}, we also have
\begin{equation*}
T^{1-2\alpha} \beta^T \to (K\lambda)^{-1}\mu^* .
\end{equation*}
In our framework, it is natural that the average size of the metaorder also scales as $T \beta^T$. Thus, following \cite{jusselin2020noarbitrage}, we consider the following assumption:
\begin{assumption}
\label{assumption:f}
There exists a constant $\gamma$ and a function $f:[0,\infty) \mapsto [0,\infty)$ satisfying $\int_0^1 f(t) \, dt = 1$ and $f(t) = 0$ for $t > 1$ such that for any $t \geq 0$, we have
\begin{equation*}
\nu^T(t) = \gamma \beta^T f(t/T).
\end{equation*}
\end{assumption}
In this assumption, $\gamma / (1+\gamma)$ is a proxy for the participation rate since the size of the metaorder is $\gamma$ times the average size of the market on a period of length $T$.

As for $N^{o,T}$, it is also natural that the average number of trades from $N^{m,T}$ on $[0,T]$ scales as $T \beta^T$. By Definition \ref{def:hawkesmarketimpact} and Equation \eqref{eq:signal}, we have
\begin{equation*}
\mathbb{E}[N^{m,T}_T] = \int_0^T \mathbb{E}\Big[\MINL^T \Big( 
\kappa \int_0^s 
\zeta^T(s-u) \, d (N^{o,T}_u - N^{m,T}_u) \Big) \Big] \, ds.
\end{equation*}
Since $N^{o,T}$ and $N^{m,T}$ scale as $T \beta^T$, it is natural that $\int_0^s 
\zeta^T(s-u) \, d (N^{o,T}_u - N^{m,T}_u)$ also scale as  $T \beta^T$. We will see in Section \ref{sec:proof:market_impact} that this is indeed the case. Therefore, we impose the following assumption
\begin{assumption}
\label{assumption:scaling_H}
There exists a locally Lipschitz increasing function $\MINL$ such that for any $x \in \mathbb{R}$, we have
\begin{equation*}
\MINL^T(x) = \beta^T \MINL\big ( (T\beta^T)^{-1} x \big).
\end{equation*}
\end{assumption}

Note that in the case where $\MINL(x) = cx^2$, this would write $\MINL^T(x) = c T^{-2}(\beta^T)^{-1} x^2$. With these assumptions in force, we can now state our main result. We introduce the rescaled market impact for $t \geq 0$ by
\begin{equation}
\label{eq:scaling:mi}
\overline{MI}^T_t 
=
\frac{1}{T\beta^T} {MI}^T_{tT}.
\end{equation}
As explained previously, the size of the metaorder scales as $T\beta^T$ and the permanent impact should be linear in the volume. Therefore, we also need to rescale the metaorder by a factor $T\beta^T$ to identify different finer behaviours. An alternative perspective would be to consider the asymptotic $\kappa \to 0$. In that case, the permanent impact would disappear and we would only observe the transient features of market impact.

\begin{theorem}
\label{thm:market_impact}
Let $I$ be a closed interval of $[0, \infty]$. Then, for $t\in I$, we have
\begin{equation}
\label{eq:limit:mi}
\sup_{t \in I}
\Big|
\overline{MI}^T_t 
-
\kappa 
\int_0^t (1 + \lambda^{-1} (t-s)^{-\alpha}) (\gamma f(s) - r^*(s)) \, ds
\Big|
\to 0.
\end{equation}
in probability, where $r^*$ is the unique solution of the integral-convolution equation
\begin{equation}
\label{eq:def:r}
r^*(t) = \MINL \Big( \kappa \int_0^t \lambda^{-1} (t-s)^{-\alpha}(\gamma f(s) - r^*(s)) \, ds \Big).
\end{equation}
\end{theorem}
In the next sections, we study in detail this result and the scaling limit identified in \eqref{eq:limit:mi}. 

\section{Study of the macroscopic limit of the market impact}
\label{sec:macroscopic:limit}

In this section, we study the limit of the rescaled market impact appearing in \eqref{eq:limit:mi}. Our goal is to understand how the market impact is influenced by the participation rate $\gamma$ and the treated volume $Q_t = \gamma \int_0^t f(s)\, ds$ in this context. We fix $f$ and set $MI(\gamma, t) = MI(f; \, \gamma, t)$ the macroscopic limit of the market impact of large metaorder obtained through Theorem \ref{thm:market_impact}. This notation differs from \eqref{eq:market_impact_volume} as we drop the dependence in $Q_t$ for the dependence in $\gamma$ since we can identify $Q_t = \gamma \int_0^t f(s)\, ds$, see also the comment after Definition \ref{def:marketimpact}.  More precisely, we write
\begin{equation}
\label{eq:market:impact}
MI(\gamma, t) = \kappa \int_0^t (1 + \lambda^{-1} (t-s)^{-\alpha}) (\gamma f(s) - r^*(\gamma, s)) \, ds
\end{equation}
where $r^*$ is the unique solution of the integral-convolution equation
\begin{equation*}
r^*(\gamma, t) = \MINL \Big( \kappa \int_0^t \lambda^{-1} (t-s)^{-\alpha}(\gamma f(s) - r^*(\gamma, s)) \, ds \Big).
\end{equation*}

Equation \eqref{eq:market:impact} can be split into three intuitive parts. First, $\kappa$ gouverns the scale of the permanent price impact of a single order, as discussed in Section \ref{sec:marketimpactsmart}. Then, the term $\gamma f(s) - r^*(\gamma, s)$ governs the dependence of the market impact in the participation rate. The dependence in the time variable $s$ is $\gamma f(s) - r^*(\gamma, s)$. We can have an explicit form for this dependence, at least in the case when $f(s) = 1$ for all $0 \leq s \leq 1$. In that case, we have from Lemma~\ref{lem:asymptotic}
\begin{equation}
\label{eq:approximationfr}
\gamma f(t) - r^*(\gamma, t)
\approx
\frac{\lambda\MINL^{-1}(\gamma)}{\kappa \Gamma(1-\alpha) \Gamma(\alpha)} \frac{1}{t^{1-\alpha}}.
\end{equation}
Eventually, the kernel of $(1 + \lambda^{-1} (t-s)^{-\alpha})$ governs the temporal aspects of the market impact. In the following, we discuss in more details various asymptotic settings.

\subsection{Temporal structure of market impact}

We first study the temporal structure of the market impact, linking our results to \cite{jusselin2020noarbitrage}. As explained in the introduction, the market impact can be decomposed into two parts: the permanent market impact $PMI$ given by
\begin{equation*}
PMI(f; \, \gamma) = \lim_{t\to\infty} MI(f; \, \gamma, t)
\end{equation*}
and the transient market impact $TMI$ given by
\begin{equation*}
TMI(f; \, \gamma, t) = MI(f; \, \gamma, t) - PMI(f; \, \gamma).
\end{equation*}
In this model, under some technical conditions on $r$, we have
\begin{equation*}
\begin{cases}
PMI(f; \, \gamma) =
\kappa \int_0^\infty (\gamma f(s) - r^*(\gamma, s)) \, ds,
\\
TMI(f; \, \gamma, t) =
\kappa \int_0^t \lambda^{-1} (t-s)^{-\alpha} (\gamma f(s) - r^*(\gamma, s)) \, ds
-
\kappa \int_t^\infty (\gamma f(s) - r^*(\gamma, s)) \, ds.
\end{cases}
\end{equation*}
Unsurprisingly, these equations show that the decay of the market impact is essentially a power-law with exponent $\alpha$, see Figure \ref{fig:power_law_market_impact}, which is consistent with empirical studies \cite{bershova2013non, gatheral2010no, moro2009market}. Similar behaviour is obtained by the Hawkes order flow model presented in \cite{jusselin2020noarbitrage}. We do not provide a rigorous proof of this statement as it is almost the same as in \cite{jusselin2020noarbitrage}, even though the additional presence of $(\gamma f(s) - r^*(\gamma, s))$ requires a careful treatment.\\

\begin{figure}[!h]
\begin{center}
\includegraphics[scale=0.35]{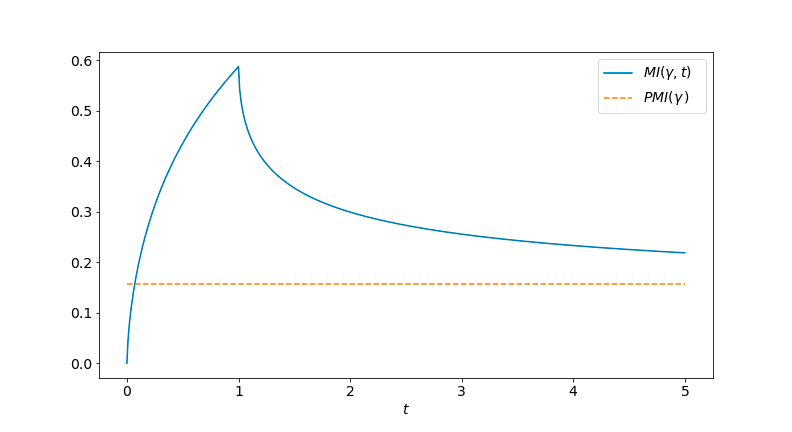}
\caption{Market impact profile when $\alpha = 0.5$, $\lambda = 1$, $\gamma = 0.3$ and $\MINL(x) = x^2$.
}
\label{fig:power_law_market_impact}
\end{center}
\end{figure}
Since $f$ is supported on $[0,1]$ by Assumption \ref{assumption:f}, we can see a decay in the price impact starting at time $1$. Note that the shape presented in Figure \ref{fig:power_law_market_impact} and the one in Figure 1 in \cite{jusselin2020noarbitrage} are very similar. However, a more precise analysis shows they slightly differ due to the presence of the term $\gamma f(s) - r^*(\gamma, s)$. We will see in the next section that this term induces a linear dependence in $\gamma$ for small $\gamma$ and a square root dependence in $\gamma$ for large $\gamma$. This switch in regime is coming from the fact the metaorder is detected and the sophisticated participant starts reacting to this metaorder. Therefore, this effect does not show up for small times and we have $\gamma f(s) - r^*(\gamma, s) \approx \gamma f(s)$ for small $s$. For large $s$, we cannot ignore $r^*$ anymore which slightly alters the shape of the time dependence of the market impact.

\subsection{Dependence on the participation rate}

The main novelty of our approach compared to \cite{jusselin2020noarbitrage} is that we can study more adequately the role of $\gamma$. We have the following result:

\begin{theorem}
\label{thm:concave_mi}
Suppose that $f(t) \geq 0$ for all $t$. Then 
\begin{itemize}
\item If $\MINL$ is increasing, then $\gamma \mapsto MI(\gamma, t)$ is increasing, for all $t \geq 0$.
\item If $\MINL$ is increasing and convex, then $\gamma \mapsto MI(\gamma, t)$ is concave, for all $t \geq 0$.
\end{itemize}
\end{theorem}

The proof of this Theorem can be found in Appendix \ref{sec:proof:concave_mi}. It relies on a comparison theorem for non-linear Volterra equations presented in Section \ref{sec:non_linear:comparison}. Note that the concavity is obtained under mild assumptions on the function $\MINL$, already discussed in Section \ref{sec:marketimpactsmart}.\\

As explained in the introduction, most studies agree that the dependence of the market impact in the participation rate should behave like a square root, at least for large enough $\gamma$ \cite{almgren2005direct, hopman2003essays, kyle2023large, lillo2003master, moro2009market}. We can obtain this shape by specifying a more precise shape of the function $\MINL$.

\begin{theorem}
\label{thm:sqrt_mi}
Suppose that $f(t) = 1$ on $[0,1]$ and that $\MINL(x) = c x^\beta$ for some $\beta > 1$ and some $c > 0$. Then for all $0 < t < 1$, we have
\begin{equation*}
MI(\gamma, t) 
\sim_{\gamma \to \infty}
\gamma^{1/\beta} c^{-1/\beta}
\Big(
1
+
\frac{\lambda}{\alpha \Gamma(1-\alpha) \Gamma(\alpha)} t^{\alpha}
\Big).
\end{equation*}
\end{theorem}

In this result, the limit $\gamma \to \infty$ corresponds to the large participation asyptotic since the participation rate is given by $\gamma / (1+\gamma)$. Notice that taking $\beta = 2$ yields the celebrated square root law. The proof of this theorem is relegated to Section \ref{sec:proof:MI_sqrt}. It is split into two steps. We first use explicitly the shape $\MINL(x) = x^\beta$ of the utility function to prove that
$r^*(\gamma, t) = \gamma  r^*(1, \gamma^{-\upsilon} t) $ for some $\upsilon = \upsilon(\beta)$. Then we study more precisely the long-term behaviour of $r^*(1, t)$ and we deduce the precise shape of $r^*$. Although the result of Theorem \ref{thm:sqrt_mi} relies on the assumption $\gamma \to \infty$, we can see on Figure \ref{fig:power_law_participation_rate} that the power law dependence of the market impact is satisfied even for relatively small values of $\gamma$.\\

\begin{figure}[!h]
\begin{center}
\includegraphics[scale=0.35]{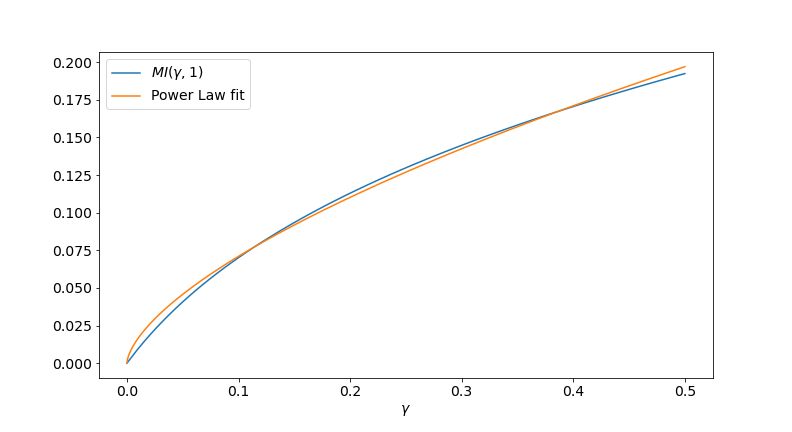}
\caption{
Market impact when $\alpha = 0.5$, $\lambda = 1$, and $\MINL(x) = x^2$, with respect to $\gamma$. A power law fit of this curve with the equation $MI = 0.3058 \, \gamma^{0.6341}$ is displayed.
}
\label{fig:power_law_participation_rate}
\end{center}
\end{figure}


\section{Numerical approximations of the market impact}

\subsection{Explicit bounds for the market impact}

From a practitioner's point of view, a very important question is the role of the choice of the trading strategy, represented in this paper by the function $f$. The effective price paid to execute a buy metaorder as presented in Section \ref{sec:marketimpactsmart} is $\int_0^T P_t \, dN^{o}_t$. The cost of the trading strategy with respect to the arrival price is therefore 
\begin{align*}
\int_0^T P_t \, dN^{o}_t - P_0 N^{o}_T
&=
\kappa 
\int_0^T
\int_0^t 
\xi(t-s) \, d (N^{a}_s + N^{o}_s - N^{b}_s - N^{m}_s)
\, dN^{o}_t
\\
&=
\kappa 
\int_0^T
\int_0^t 
\xi(t-s) \, d (N^{a}_s - N^{b}_s)
\, dN^{o}_t
+
\int_0^T
MI_t
\, dN^{o}_t.
\end{align*}
Since $\int_0^t \xi(t-s) \, d (N^{a}_s - N^{b}_s)$ is a martingale price, it is natural to focus on $\int_0^T MI_t \, dN^{o}_t$ and to try to minimize it by choosing an optimal $f$. This quantity is random so one can consider its expectation. However, as discussed after Definition \ref{def:marketimpact}, the computation is intricate. More precisely, we have
\begin{equation*}
\mathbb{E}[MI_t | N^o] 
= \int_0^t  \mathbb{E}[\lambda^{o}_s - \lambda^{m}_s | N^o] \, ds
= \int_0^t  \gamma f(s) - \mathbb{E}[ \MINL \big( P_{s-} - EP_{s-} \big)  | N^o] \, ds.
\end{equation*}
Carrying this computation further is not easy as it requires to understand the convexity adjustment needed to correct for $\mathbb{E}[\MINL(N^o-N^m)] \neq \MINL(\mathbb{E}[N^o-N^m])$. Nevertheless, two routes are available. The first is to use Monte-Carlo simulations to compute this expectation and deduce the cost of the underlying execution strategy. The second one, which we detail below, consists in replacing $\int_0^T MI_t \, dN^{o}_t$ by its scaling limit counterpart. From Theorem \ref{thm:market_impact},  $\int_0^T MI_t \, dN^{o}_t$ can be replaced by $\int_0^T MI(\gamma, t) \gamma f(t) \, dt$ (up to a scaling constant that does not play any role in the subsequent analysis), where 
\begin{equation*}
MI(\gamma, t) = \kappa \int_0^t (1 + \lambda^{-1} (t-s)^{-\alpha}) (\gamma f(s) - r^*(\gamma, s)) \, ds
\end{equation*}
where $r^*$ is the unique solution of the integral-convolution equation
\begin{equation}
\label{eq:def:r_approx}
r^*(\gamma, t) = \MINL \Big( \kappa \int_0^t \lambda^{-1} (t-s)^{-\alpha}(\gamma f(s) - r^*(\gamma, s)) \, ds \Big).
\end{equation}

The main difficulty in this approach comes from the lack of closed form expression for $r^*$. If $\MINL$ were linear, we could use Mittag-Leffler function that often appears to derive the explicit form of solutions of \eqref{eq:def:r}, see \cite{callegaro2021fast, eleuch2019characteristic, gatheral2019rational}. The linear case is in fact used in Section \ref{sec:proof:MI_sqrt} in conjunction with comparison theorems to prove Theorem \ref{thm:sqrt_mi}. However, when $\MINL$ is not linear, we must rely on numerical schemes. The most intuitive and earliest approach is to implement the Euler scheme \cite{cryer1972numerical}, and it is the one we use in our simulations (Figures \ref{fig:power_law_market_impact} and \ref{fig:power_law_participation_rate}). Beyond that, more elaborate techniques have been developed for non-linear Volterra equation, such as the Laplace Transform Method \cite{li2023asymptotic},
Galerkin Method \cite{nedaiasl2021galerkin},
piecewise collocation \cite{cardone2020stability} or iterative reproducing kernel Hilbert spaces method \cite{sakar2017solutions}. \\

Using the structure of \eqref{eq:def:r_approx}, we can also use comparison theorems for Volterra integral differential equations to derive some explicit and exact bounds on the market impact when $\gamma$ is small. These bounds can then be used as proxy for the market impact to get bounds on the cost of a given strategy.

\begin{theorem}
\label{thm:small:gamma}
Suppose that $\MINL$ is increasing. Then for any $\gamma$ and any $t$, $MI_-(\gamma, t) \leq MI(\gamma, t) \leq MI_+(\gamma,t)$, where we write
\begin{equation*}
MI_\pm(\gamma, t) =
\kappa 
\int_0^t (1 + \lambda^{-1} (t-s)^{-\alpha}) (\gamma f(s) - r^*_\mp(s)
)\, ds
\end{equation*}
with
\begin{equation*}
r^*_+(\gamma, s) = \MINL(\gamma \kappa F_{\alpha, \lambda}(s))
\;\;\;\text{ and } \;\;\;r^*_-(\gamma, s) = \MINL\Big(\gamma \kappa  F_{\alpha, \lambda}(s) - \kappa  \int_0^s \lambda^{-1} (s-u)^{-\alpha} \MINL(
 \gamma \kappa F_{\alpha, \lambda}(u) ) \, du\Big)
\end{equation*}
and where  $F_{\alpha, \lambda}(t) =  \lambda^{-1} \int_0^t (t-s)^{-\alpha} f(s) \, ds$.
\end{theorem}  

For reasonable choices of $\MINL$, the integral term appearing in $r^*_-$ should be small compared to $\gamma F_{\alpha, \lambda}(s)$. We illustrate this in the special case $\MINL(x) = x^2$ where the computation can be made explicit.

\begin{corollary}
Suppose that $\MINL(x) = x^2$. Then we have
\begin{equation*}
r^*_+(\gamma, s) = \gamma^2 \kappa^2 (\lambda^{-1} F_{\alpha, \lambda}(s))^2
\;\;\;\;\text{ and }\;\;\;\;
r^*_-(\gamma, s) = 
\gamma^2 \kappa^2 (\lambda^{-1} F_{\alpha, \lambda}(s))^2 + o(\gamma^3)
\end{equation*}
yielding to
\begin{equation*}
MI(\gamma, t) = 
\gamma \kappa
\int_0^t (1 + \lambda^{-1} (t-s)^{-\alpha})  f(s) \, ds
-
\gamma^2 \kappa^3 \int_0^t (1 + \lambda^{-1} (t-s)^{-\alpha})  (\lambda^{-1} F_{\alpha, \lambda}(s))^2 \, ds
+ o(\gamma^2).
\end{equation*}
\end{corollary}

In general Theorem \ref{thm:small:gamma} suggests that an appropriate approximation of the scaling limit of market impact for small $\gamma$ would be
\begin{equation}
\label{eq:approx:mi}
MI(\gamma, t) \approx \kappa
\int_0^t (1 + \lambda^{-1} (t-s)^{-\alpha}) (\gamma f(s) - \MINL(\gamma \kappa F_{\alpha, \lambda}(s))
)\, ds.
\end{equation}

\subsection{Adomian decomposition}

The method used in Theorem \ref{thm:small:gamma} can be improved using an Adomian decomposition method inspired from \cite{hu2008analytical, khan2013fractional}. Suppose that $\MINL(x) = \sum_{k\geq 1} \MINL^{(k)}(0) x^k / k!$ on $[0, \infty)$. Then, we introduce $u_1, \dots, u_J$ some functions on $[0, \infty)$ such that
\begin{align*}
u_1(t) + \kappa \MINL'(0) \int_0^t \lambda^{-1} (t-s)^{-\alpha} u_1(s)\, ds =  \kappa F_{\alpha,\lambda}(t)
\end{align*}
and, for $2 \leq l \leq J$
\begin{align*}
u_l(t) + \kappa \MINL'(0) \int_0^t \lambda^{-1} (t-s)^{-\alpha} u_l(s)\, ds = - \sum_{j=2}^{l} \sum_{k_1 + \dots + k_j = l} \kappa^j \frac{\MINL^{(j)}(0)}{j!} \int_0^t \lambda^{-1} (t-s)^{-\alpha} \prod_{\ell=1}^j u_{k_\ell}(s)\, ds. 
\end{align*}
Such function always exists: when $\MINL'(0) \neq 0$, these implicit definitions require to solve a linear Volterra convolution equation and when $\MINL'(0) = 0$, the $(u_l)_l$ are explicitly defined by induction. Then we write $u_{(J)}(\gamma, t) = \sum_{l=1}^{J} \gamma^l u_l(t)$ and we have
\begin{align*}
MI(\gamma, t) 
&= \kappa \int_0^t (1 + \lambda^{-1} (t-s)^{-\alpha}) (\gamma f(s) - \MINL(u_{(J)}(\gamma, s)
)\, ds + o(\gamma^J)
\\
&= \kappa \int_0^t (1 + \lambda^{-1} (t-s)^{-\alpha}) \Big(\gamma f(s) - 
\sum_{l=1}^J \gamma^l \sum_{j=0}^J \frac{\MINL^{(j)}(0) }{j!} \sum_{k_1+\dots+k_j=l}
\, \prod_{\ell=1}^j u_{k_\ell}(s) \Big)ds + o(\gamma^J).
\end{align*}
Note that $\kappa \int_0^t (1 + \lambda^{-1} (t-s)^{-\alpha}) (\gamma f(s) - \MINL(u_{(J)}(\gamma, s)
)\, ds$ corresponds to $MI_-(\gamma,t)$ and $MI_+(\gamma,t)$ from Theorem \ref{thm:small:gamma} in the cases $J=1$ and $J=2$ respectively. Moreover, when $\MINL(x) = cx^2$, these equations can be simplified and we get $u_1(t) = \lambda^{-1} F_{\alpha, \lambda}(t)$ and 
\begin{align*}
u_l(t) = - \sum_{k = 1}^{l-1} \kappa^2c\lambda^{-1} \int_0^t (t-s)^{-\alpha} u_{k}(s)u_{l-k}(s) \, ds. 
\end{align*}
The effectiveness of this procedure when $\gamma$ is not too large is illustrated in Figure \ref{figure:approx}. We refer to \cite{khan2013fractional} for more details about the numerical implementation, and also for the use of Padé approximation to reduce the error for large $t$. 

\begin{figure}[ht]
\label{figure:approx}
\centering
  \includegraphics[width=0.48\linewidth]{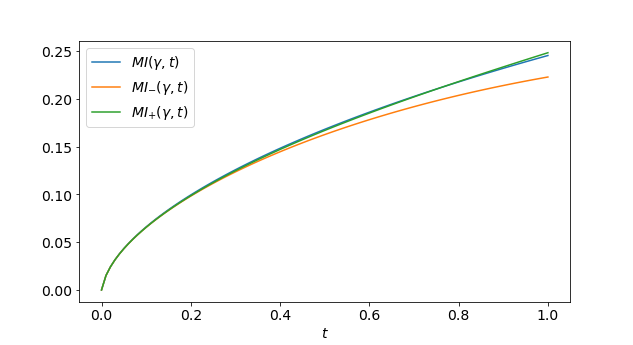}
  \includegraphics[width=0.48\linewidth]{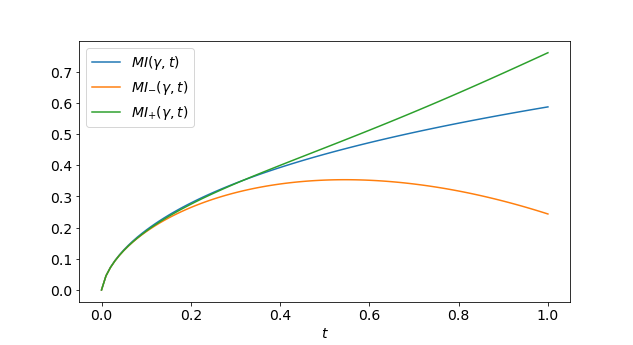}
  \caption{Market impact profiles and the exact bounds presented in Theorem \ref{thm:small:gamma}, with $\alpha = 0.5$, $\lambda = 1$, $\MINL (x) = x^2$. The left figure corresponds to $\gamma = 0.1$ while the right one corresponds to $\gamma=0.3$.}
\end{figure}

\section*{Acknowledgments}
Mathieu Rosenbaum and Gr\'egoire Szymanski gratefully acknowledge the financial support of the \'Ecole Polytechnique chairs {\it Deep Finance and Statistics} and {\it Machine Learning and Systematic Methods}.

\bibliographystyle{alpha}
\bibliography{library}

\appendix

\section{Proof of Theorem \ref{thm:whole_model}}
\label{sec:proof:whole_model}

We start by proving point \ref*{thm:whole_model:ii} of Theorem \ref{thm:whole_model}.  This proof is an adaptation of Theorem 6 in \cite{delattre2016hawkes} to our specific setup. Suppose that $(\pi^a, \pi^b, \pi^o, \pi^m)$ are independent Poisson point measures on $[0,\infty) \times [0, \infty)$ adapted to the filtration $(\mathcal{F}_t)_t$. From Theorem 6 in \cite{delattre2016hawkes}, we already know that exists a pathwise unique family $(N^a, N^b, N^o)$ satisfying Equation \eqref{eq:def:Npoint} and, moreover, $N^a$ and $N^b$ are independent Hawkes processes with baseline $\mu$ and self-exciting kernel $\varphi$ and $N^o$ is an inhomogeneous Poisson process with intensity $\lambda^o_t$. It remains to build adequately $N^m$. We first show uniqueness. Suppose $N^m$ and $\widetilde{N}^m$ are solution of 
\begin{align}
\label{eq:def:double:N}
N^m_t = \int_0^t \int_0^\infty \1_{z \leq \lambda^m_s} \, \pi^m(ds\,dz) \;\;\text{ and }\;\;
\widetilde{N}^m_t = \int_0^t \int_0^\infty \1_{z \leq \widetilde{\lambda}^m_s} \, \pi^m(ds\,dz)
\end{align}
where the intensities $\lambda^m$ and $\widetilde{\lambda}^m$ are given by
\begin{align*}
\lambda^m_t = \MINL \big( P_{t-} - EP_{t-} \big)
\;\;\text{ and }\;\;
\widetilde{\lambda}^m_t = \MINL \big( \widetilde{P}_{t-} - \widetilde{EP}_{t-} \big)
\end{align*}
and where we have
\begin{align*}
\begin{cases}
P_t 
=
P_0 + \kappa 
\int_0^t 
\xi(t-s) \, d (N^{a}_s + N^{o}_s - N^{b}_s - N^{m}_s),
\\
EP_t 
=
P_0 + \kappa 
\int_0^t 
\xi(t-s) \, d (N^{a}_s  - N^{b}_s) + \kappa(N^{o}_t - N^{m}_t);
\\
\widetilde{P}_t 
=
P_0 + \kappa 
\int_0^t 
\xi(t-s) \, d (N^{a}_s + N^{o}_s - N^{b}_s - \widetilde{N}^{m}_s),
\\
\widetilde{EP}_t 
=
P_0 + \kappa 
\int_0^t 
\xi(t-s) \, d (N^{a}_s  - N^{b}_s) + \kappa(N^{o}_t - \widetilde{N}^{m}_t).
\end{cases}
\end{align*}
We aim at proving that almost surely, $N^m = \widetilde{N}^m$. To do so, we introduce $\Delta_t = \int_0^t |d (N^{m}_s - \widetilde{N}^{m}_s)|$ and $\delta_t = \mathbb{E}[\Delta_t | N^o]$ for all $t\geq 0$ (which is well defined because $\Delta_t \geq 0$) and we prove that $\delta_t = 0$ so that $\Delta_t = 0$ almost surely for all $t \geq 0$. By Equation \eqref{eq:def:double:N}, we have
\begin{align*}
\Delta_t = 
\int_0^t \int_0^\infty  \Big |\1_{z \leq \lambda^m_s} - \1_{z \leq \widetilde{\lambda}^m_s} \Big| \, \pi^m(ds\,dz).
\end{align*}
Taking conditional expectation with respect to $N^o$, which is independent of $\pi^m$ by construction, we get
\begin{align*}
\delta_t 
&= 
\int_0^t \int_0^\infty  \mathbb{E}\Big[  \big |\1_{z \leq \lambda^m_s} - \1_{z \leq \widetilde{\lambda}^m_s} \big| \, \Big | \, N^o  \Big] \,dz\, ds
=
\int_0^t  \mathbb{E}\Big[  \big |\lambda^m_s - \widetilde{\lambda}^m_s\big| \, \Big | \, N^o \Big] \, ds.
\end{align*}
Moreover, we have
\begin{align*}
\lambda^m_t - \widetilde{\lambda}^m_t
& =
\MINL \big( P_{t-} - EP_{t-} \big)
-
\MINL \big( \widetilde{P}_{t-} - \widetilde{EP}_{t-} \big)
\\
&=
\MINL \Big( \kappa 
\int_0^{t-} 
\zeta(t-s) \, d (N^{o}_s - N^{m}_s)
 \Big)
-
\MINL \Big( 
\kappa 
\int_0^{t-}
\zeta(t-s) \, d ({N}^{o}_s - \widetilde{N}^{m}_s)
 \Big).
\end{align*}
Let $T > 0$ and $0 \leq t \leq T$. Since $\zeta$ is a non-negative bounded function, we know that $\int_0^{t- } 
\zeta(t-s) \, d (N^{o}_s - N^{m}_s) \leq \norm{\zeta}_\infty N^{o}_T$ and $\int_0^{t- } 
\zeta(t-s) \, d (N^{o}_s - \widetilde{N}^{m}_s) \leq \norm{\zeta}_\infty N^{o}_T$ for any $t\leq T$. By definition, for each $M > 0$, $\MINL$ is Lipschitz on $(-\infty, \kappa \norm{\zeta}_\infty M]$ with Lipschitz constant $L_M$ on this interval. Thus we have
\begin{align*}
\delta_t 
&
=
\int_0^t  \mathbb{E}\Big[  \big |\lambda^m_s - \widetilde{\lambda}^m_s\big| \, \Big | \, N^o \Big] \, ds
\\
&
\leq
\int_0^t  \mathbb{E}\Big[  L_{N^o_T} \big| \kappa 
\int_0^{s-} 
\zeta(s-u) \, d (N^{o}_u - N^{m}_u)
-
\kappa 
\int_0^{s-} 
\zeta(s-u) \, d (N^{o}_u - \widetilde{N}^{m}_u)
\big|
\, \Big | \, N^o
\Big] \, ds
\\
&
\leq
 L_{N^o_T}
\kappa
\int_0^t  \mathbb{E}\Big[   \big| 
\int_0^{s-} 
\zeta(s-u) \, d (\widetilde{N}^{m}_u - N^{m}_u)
\big|
\, \Big | \, N^o
\Big] \, ds
\\
&
\leq
 L_{N^o_T}
\kappa
\int_0^t  \mathbb{E}\Big[
\int_0^{s-} 
\zeta(s-u) \, d \Delta_u
\, \Big | \, N^o
\Big] \, ds
\\
&
\leq
 L_{N^o_T}
\kappa
\norm{\zeta}_\infty
\int_0^t \delta_s \, ds.
\end{align*}
Using Gronwall's lemma, we get $\delta_t = 0$ which concludes uniqueness.\\

It remains to prove existence of $N^m$, which is done through Picard iterations. We first define $N^{m,0}_t = 0$ for all $t$ and then for $k \geq 0$
\begin{align*}
\begin{cases}
P^k_t 
=
P_0 + \kappa 
\int_0^t 
\xi(t-s) \, d (N^{a}_s + N^{o}_s - N^{b}_s - N^{m,k}_s),
\\
EP_t^k 
=
P_0 + \kappa 
\int_0^t 
\xi(t-s) \, d (N^{a}_s  - N^{b}_s) + \kappa(N^{o}_t - N^{m,k}_t);
\end{cases}
\end{align*}
and
\begin{align}
\label{eq:def:double:Nk}
N^{m,k+1}_t = \int_0^t \int_0^\infty \1_{z \leq \lambda^{m,k+1}_s} \, \pi^m(ds\,dz) \;\;\text{ where }\;\;
\lambda^{m,k+1}_t = \MINL \big( P^k_{t-} - EP^k_{t-} \big).
\end{align}
We also define
\begin{align*}
\Delta_t^{k} = \int_0^t \, |d (N^{m,k+1}_s -  {N}^{m,k}_s)|\;\;\text{ and }\;\;
\delta_t^{k} = \mathbb{E}[\Delta_t^{k} \, | \, N^o].
\end{align*}
Repeating the same computations as for the uniqueness part, we check that for any $\varepsilon$ and $M$ chosen as before, we have for any $k \geq 0$
\begin{align*}
\delta_t^{k+1}
\leq
L_{N^o_T}
\kappa
\norm{\zeta}_\infty
\int_0^t \delta_s^{k} \, ds.
\end{align*}
By Gronwall inequality for sequences of functions (See for instance Lemma 23 in \cite{delattre2016hawkes} with $\varphi$ constant), this implies that $\sum_{k=0}^\infty \delta_t^{k}$ is bounded on $[0,T]$ by a finite random variable $\sigma(N^o)$-measurable depending only on $\delta^0$ and $L_{N^o_T} \kappa$. This implies that the Picard sequence $(N^{m,k})_k$ is Cauchy and therefore converges towards a random process $Z$ such that
\begin{align*}
\int_0^t \, |d (N^{m,k}_s -  Z_s)| \to 0.
\end{align*} 
From the definitions of $P^k_t$ and $EP^k_t$, we deduce that
\begin{align*}
\begin{cases}
P^k_t 
\to
P_0 + \kappa 
\int_0^t 
\xi(t-s) \, d (N^{a}_s + N^{o}_s - N^{b}_s - Z_s),
\\
EP_t^k 
\to
P_0 + \kappa 
\int_0^t 
\xi(t-s) \, d (N^{a}_s  - N^{b}_s) + \kappa(N^{o}_t - Z_t)
\end{cases}
\end{align*}
and from Equation \eqref{eq:def:double:Nk} that
\begin{align*}
Z_t = \int_0^t \int_0^\infty \1_{z \leq \lambda^{Z}_s} \, \pi^m(ds\,dz) \;\;\text{ where }\;\;
\lambda^{Z}_t = \MINL \Big( \kappa 
\int_0^t 
\zeta(t-s) \, d (N^{o}_s - Z_s)
 \Big).
\end{align*}
\\

We now prove \ref*{thm:whole_model:i} of Theorem  \ref{thm:whole_model}. The proof is technical but classical. We use the work of Brémaud and Massoulié \cite{bremaud1996stability}, which, relies on the results of Jacod \cite{jacod1979calcul} (Chapter $14$, pages $469$ to $478$). Additionally, we refer to the comprehensive proof presented by Chevallier \cite{chevallier2013detection}. Nevertheless, we recall here the main ideas of the construction.\\

Suppose that $(N^a, N^b, N^o, N^m)$ is a Hawkes market impact model with a sophisticated trader as defined in Definition \ref{def:hawkesmarketimpact}. Up to an extension of the probability space, we assume the existence of independent Poisson random measures $(\widetilde{\pi}^a,\widetilde{\pi}^b,\widetilde{\pi}^o,\widetilde{\pi}^m)$ on $[0,T] \times [0, \infty)$ with intensity $ds \, dz$ independent of $(N^a, N^b, N^o, N^m)$ and, of a family of i.i.d. random variables $((U^a_j, U^b_j, U^o_j, U^m_j))_{j}$ independent of the aforementioned variables and uniformly distributed on $[0,1]$. We also write 
$((T^a_{j},T^b_{j},T^o_{j},T^m_{j}))_{j}$ for the jump times of the processes $(N^a, N^b, N^o, N^m)$, \textit{i.e.} the variables such that $N^x_t= \sum_{j \geq 1} \1_{T^x_{j} \leq t}$ for any $x \in \{a,b,o,m\}$. We define then
\begin{align*}
\pi^x (ds \, dz) = \sum_{j=1}^{\infty} \delta_{(T^x_{j}, \lambda^x_{T^x_{j}} U^x_{j})}(ds \, dz) + \1_{z > \lambda^x_s}\, \widetilde{\pi}^x(ds \, dz).
\end{align*}
With this, we first check that \eqref{eq:def:Npoint} holds. Since $ \1_{z \leq \overline{Y}(s)}\1_{z > \overline{Y}(s)} = 0$, we have
\begin{align*}
\int_0^t \int_0^\infty \1_{z \leq \lambda^x_s} \, \pi^x(ds\,dz)
&=
\sum_{j=1}^{\infty}
\int_0^t \int_0^\infty \1_{z \leq \lambda^x_s} \, \delta_{(T^x_{j}, \lambda^x_{T^x_{j}} U^x_{j})}(ds \, dz)
+
\int_0^t \int_0^\infty \1_{z \leq \lambda^x_s} \1_{z > \lambda^x_s} \, \widetilde{\pi}^x(ds \, dz)
\\
&=
\sum_{j=1}^{\infty}
\int_0^t \int_0^\infty \1_{z \leq \lambda^x_s} \, \delta_{(T^x_{j}, \lambda^x_{T^x_{j}} U^x_{j})}(ds \, dz)
\\
&=
\sum_{j=1}^{N^x_t}
\1_{ \lambda^x_{T^x_{j}} U^x_{j} \leq \lambda^x_{T^x_{j}}} = N^x_t.
\end{align*}
It remains to prove that the $({\pi}^a,{\pi}^b,{\pi}^o,{\pi}^m)$ are independent random Poisson measures on $[0,\infty) \times [0, \infty)$ with intensity $ds \, dz$. This can be achieved by adapting \cite{jacod1979calcul}, Chapter $14$, pages $469$ to $478$ to the multidimensional setting. We skip the details for conciseness.\\

\section{Proof of Theorem \ref{thm:market_impact}}
\label{sec:proof:market_impact}

\subsection{Outline of the proof}
\label{sec:outline}

Without loss of generality, we suppose in this proof that $I = [0, M]$ for some $M > 1$. By rescaling the price, we may also suppose that $\kappa = 1$. For later use, we introduce few notations. We write
\begin{align*}
\begin{cases}
\overline{N}^{o,T}_t = (T\beta^T)^{-1} N^{o,T}_{tT},
\\
\overline{M}^{o,T}_t = (T\beta^T)^{1/2} \big( \overline{N}^{o,T}_t - \gamma \int_0^t f(s) \, ds\big),
\end{cases}
\;\;\;\;\text{ and }\;\;\;\;
\begin{cases}
\overline{N}^{m,T}_t = (T\beta^T)^{-1} N^{m,T}_{tT}
\\
\overline{M}^{m,T}_t = (T\beta^T)^{1/2} \big( \overline{N}^{m,T}_t - \int_0^t \overline{\lambda}^{m,T}(s) \, ds\big)
\end{cases}
\end{align*}
where $\overline{\lambda}^{m,T}_t := (\beta^T)^{-1} \lambda^{m,T}_{tT}$. We also set
\begin{align}
\label{eq:def:lambdatilde}
\widetilde{\lambda}^{o,T}_t = \int_0^t 
\overline{\zeta}^T(t-s) \, d \overline{N}^{o,T}_s
\;\;\text{ and }\;\;
\widetilde{\lambda}^{m,T}_t = \int_0^t 
\overline{\zeta}^T(t-s) \, d \overline{N}^{m,T}_s
\end{align}
where $\overline{\zeta}^T(t) = {\zeta}^T(tT)$ so that we have by Assumption \ref{assumption:scaling_H}
\begin{align*}
\overline{\lambda}^{m,T}_t = \MINL\Big(\widetilde{\lambda}^{o,T}_t - \widetilde{\lambda}^{m,T}_t\Big).
\end{align*}

The rest of the proof is structured as follows. First in Section \ref{sec:proof:properties_gammas}, we gather some asymptotic properties on the function $\overline{\zeta}^T$ that will be used throughout the proof. Then in Section \ref{sec:proof:uniqueness}, we show that the solution $r^*$ of Equation \eqref{eq:def:r} is unique. Therefore, the proof consists in proving that whenever $T = T_n \to \infty$, we can find a subsequence (that we still write $T=T_n \to \infty$ for conciseness) such that 
\begin{align*}
\sup_{t \leq M}
\Big|
\overline{MI}^T_t - \int_0^t (1 + \lambda^{-1} (t-s)^{-\alpha}) (\gamma f(s) - r^*(s) ) \, ds \Big|
\to 0.
\end{align*}
We need several steps to build this extraction: 
\begin{itemize}
\item \textit{Step $1$.} We first study the sequence $\overline{N}^{o,T}$  and we show that (up to a subsequence) it converges almost surely towards $\gamma F$ where $F$ is the primitive of $f$ vanishing at $0$.
\item \textit{Step $2$.} We study $\widetilde{\lambda}^{o,T}$ and we show that it converges towards a deterministic limit, uniformly on $[0,M]$.
\item \textit{Step $3$.} We study the sequence $\overline{N}^{m,T}$ and we show that (up to a subsequence) it converges almost surely towards a continuous process.
\item \textit{Step $4$.} We study $\widetilde{\lambda}^{m,T}$ and we show that it converges uniformly on $[0,M]$ towards a limit.
\item \textit{Step $5$.} We combine the previous results to prove that $\overline{\lambda}^{m,T}$ converges uniformly towards $r^*$ and we conclude.
\end{itemize}

%
%
%
%
%
%
%
%

\subsection{Properties of $\overline{\zeta}^T$}
\label{sec:proof:properties_gammas}
%

We first study the properties of the sequence of measures $\overline{\zeta}^T(s)
\, ds$ and we show in the following lemma that it converges weakly towards $\lambda^{-1} s^{-\alpha}
\, ds$.
\begin{lemma}
\label{lemma:convergence:gammameasure}
For any bounded continuous function $g : [0,M] \to \mathbb{R}$, we have
\begin{align*}
\lim\limits_{T\to\infty}
\int_0^M
g(u)
\overline{\zeta}^T(s)
\, ds
=
\int_0^M
g(u)\, 
\lambda^{-1} s^{-\alpha}
\, ds.
\end{align*}
\end{lemma}
\begin{proof}
Let $m^T$ and $m^\infty$ denote the measures on $[0,M]$ defined by
\begin{align*}
m^T(dt)=
\overline{\zeta}^T(t)
\, dt
\;\;\text{ and } \;\;
m^\infty(dt)=
\lambda^{-1} t^{-\alpha}
\, dt.
\end{align*}
Since both these measures are finite, Lemma~\ref{lemma:convergence:gammameasure} states that $m^T$ converges weakly towards $m^\infty$. To prove this, it is enough to prove pointwise convergence of their cumulative distribution functions $R^T$ and $R^\infty$ defined by
\begin{align*}
R^T(t) = \int_0^t m^T(dt) = \int_0^t
{\overline{\zeta}}^T(t)
\;\;\text{ and } \;\;
R^\infty(t) = \int_0^t m^\infty(dt) = \frac{t^{1-\alpha}}{(1-\alpha)\lambda}.
\end{align*} 
We first fix $t > 0$ and $\varepsilon > 0$. By definition, $\overline{\zeta}^T(t) = \xi^T(tT) - 1$ where $\xi^T$ is defined in Equation \eqref{eq:def:xit}. Thus we have
\begin{align}
\label{eq:expr:cumulmT}
R^T(t) = \frac{a^T}{T(1-a^T)} \int_0^{tT} \int_s^\infty \varphi(u)\,du\,ds.
\end{align}
By Assumption \ref{assumption:a}, there exists $A > 0$ such that for all $s \geq A$,
\begin{align*}
s^{-\alpha} (K - \varepsilon) \leq  \int_s^\infty \varphi(u)\,du \leq s^{-\alpha} (K + \varepsilon).
\end{align*}
Moreover, $0 \leq \int_s^\infty \varphi(u)\,du \leq 1$ for all $s \leq A$. Plugging this into \eqref{eq:expr:cumulmT}, we get
\begin{align*}
\frac{a^T}{T(1-a^T)} \int_A^{tT} \frac{K-\varepsilon}{s^\alpha}\,ds
\leq 
R^T(t) \leq \frac{a^T}{T(1-a^T)} \Big( \int_A^{tT} \frac{K-\varepsilon}{s^\alpha}\,ds + A \Big)
\end{align*}
for all $T$ such that $tT > A$. Moreover
\begin{align*}
\frac{a^T}{T(1-a^T)} \int_A^{tT} \frac{K-\varepsilon}{s^\alpha}\,ds
&=
\frac{a^T  (K-\varepsilon)}{T^\alpha (1-a^T)(1-\alpha)}\big(t^{1-\alpha} - (A/T)^{1-\alpha}\big)
\end{align*}
and by Assumption \ref{assumption:a}, we know that $T^\alpha (1-a^T) \to (1-\alpha)^{-1} K \lambda$ and $a^T \to 1$. Therefore, we obtain
\begin{align*}
\frac{a^T}{T(1-a^T)} \int_A^{tT} \frac{K-\varepsilon}{s^\alpha}\,ds
&\to
\frac{K-\varepsilon}{K \lambda}\, t^{1-\alpha}.
\end{align*}
Similarly we have
\begin{align*}
\frac{a^T}{T(1-a^T)} \Big( \int_A^{tT} \frac{K+\varepsilon}{s^\alpha}\,ds + A \Big)
&\to
\frac{K+\varepsilon}{K \lambda }\, t^{1-\alpha}.
\end{align*}
Therefore, we get
\begin{align*}
\frac{K-\varepsilon}{K \lambda}\, t^{1-\alpha}
\leq
\liminf\limits_{T \to \infty} R^T(t)
\leq
\limsup\limits_{T \to \infty} R^T(t)
\leq
\frac{K+\varepsilon}{K \lambda}\, t^{1-\alpha}.
\end{align*}
We obtain the convergence $R^T(t) \to R^\infty(t)$ by letting $\varepsilon \to 0$, which concludes the proof of Lemma~\ref{lemma:convergence:gammameasure}.
\end{proof}

We now study the convergence of convolution of $\overline{\zeta}^T(s) \, ds$ with a stochastic process.
\begin{lemma}
\label{lemma:convol:gamma}
Suppose that $Z^T$ is a family of processes and $Z$ is a process such that
\begin{enumerate}
\item There exists a piecewise constant process $X^T$ having finitely many jumps on any finite time intervals and a positive process $\alpha^T$ such that $Z^T = X^T - A^T$ where $A^T_t = \int_0^t \alpha^T_s \, ds$ for all $t \geq 0$,
\item $Z^T \to Z$ uniformly on $[0,M]$,
\item $Z$ is continuous and $Z_0 = 0$.
\end{enumerate}
Then we have 
\begin{align*}
\sup_{t \leq M} 
\Big|
(T\beta^T)^{-1/2}
\int_0^{t} 
\overline{\zeta}^T(t-s) \, d Z^T_s 
\Big|
\to 0.
\end{align*}
\end{lemma}

\begin{proof}
Recall that $\overline{\zeta}^T(t) = \xi^T(tT) - 1$ where $\xi^T$ is defined in Equation \eqref{eq:def:xit}. Therefore, we have 
\begin{align}
\label{eq:rewrite:goal}
(T\beta^T)^{-1/2}
\int_0^{t} 
\overline{\zeta}^T(t-s) \, d Z^T_s
&=
\frac{a^T}{(1-a^T)(T\beta^T)^{1/2}}
\int_0^{t} 
\int_{T(t-s)}^\infty \varphi(u) \, du \, d Z^T_s.
\end{align}
By Assumptions \ref{assumption:a}, we know that 
\begin{align*}
(1-a^T)(T\beta^T)^{1/2} = (T\mu^T(1-a^T))^{1/2} \to \Big(\frac{\mu^* \lambda K}{1-\alpha}\Big)^{1/2} > 0.
\end{align*}
Therefore, it is enough to prove that the integral on the right hand side of \eqref{eq:rewrite:goal} converges to $0$ uniformly for $0 \leq t \leq M$. Then we write $dZ^T_s = dX^T_s - \alpha^T_s \, ds$ and using that $X^T$ is a pure jump process and that $\alpha^T$ is positive, we can apply Fubini theorems (pathwise) and deduce
\begin{align*}
\int_0^{t} 
\int_{T(t-s)}^\infty \varphi(u) \, du \, d Z^T_s
 &=
 \int_0^{t} \int_{T(t-s)}^\infty \varphi(u) \, du \, d X^T_s
 +
\int_0^{t} \int_{T(t-s)}^\infty \varphi(u) \, du \, \alpha^T_s \, ds\\
&=
\int_{0}^\infty
 \varphi(u) 
 \int_{\max( t - u/T, \, 0)}^{t} 
 \, d X^T_s \, du
 +
 \int_{0}^\infty
 \varphi(u) 
 \int_{\max( t - u/T, \, 0)}^{t} 
 \alpha^T_s \, d s \, du
\\
&=
\int_{0}^\infty
 \varphi(u) 
  \big(Z^T_t
  -
   Z^T_{(t - u/T)_+}\big)
 \, du.
\end{align*}
Moreover, we have
$\sup_{t \leq M}
|
\int_{0}^\infty
 \varphi(u) 
  \big(Z^T_t
  -
   Z^T_{(t - u/T)_+}\big)
 \, du
| 
\leq I^T + II^T + III^T
$
where
\begin{align*}
I^T &= 
\sup_{t \leq M}
\Bigg |
\int_{0}^\infty
 \varphi(u) 
  \big(Z^T_t
  -
   Z^T_{(t - u/T)_+}\big)
 \, du
-
\int_{0}^\infty
 \varphi(u) 
  \big(Z_t
  -
   Z_{(t - u/T)_+}\big)
 \, du
 \Bigg |
,
\\
II^T &= 
\sup_{t \leq M}
\Big |
Z_t
\int_{tT}^\infty
 \varphi(u) 
du
\Big |,
\\
\text{ and } \;\;\;\;III^T &= 
\sup_{t \leq M}
\Big |
\int_{0}^{t}
 T\varphi(uT) 
  \big(Z_t
  -
   Z_{t - u}\big)
 \, du
 \Big |.
\end{align*}
We deal with these three terms separately. The term $I^T$ is the easiest to handle. Since $\norm{\varphi}_{L^1} = 1$ and $Z^T \to Z$ uniformly on $[0,M]$, we have $ I^T \leq 2 
\sup_{t \leq M} |Z^T_t - Z_t|
\to 0$ which converges to $0$. We now focus on $II^T$. For each $\delta > 0$, we have
\begin{align}
\label{eq:term1supxint}
\sup_{t \leq M} \Big|
Z_t
\int_{tT}^\infty
 \varphi(u) 
du
\Big|
\leq \sup_{t \leq \delta} |Z_t| + \sup_{t \leq M} |Z_t| \int_{\delta T}^\infty
 \varphi(u) 
du
\end{align}
using that  $\norm{\varphi}_{L^1} = 1$.
Since $Z$ is continuous and $Z(0)=0$, for each $\varepsilon > 0$, there exists $\delta = \delta(\varepsilon) > 0$ small enough such that $|Z_t| \leq \varepsilon/2$ for each $0 \leq t \leq \delta$. Moreover, $Z$ is bounded on $[0,M]$ since it is continuous and $\int_{\delta T}^\infty
 \varphi(u) 
du \to 0$ as $T \to \infty$. Therefore we can take $T$ big enough so that $\sup_{t \leq M} |Z_t| \int_{\delta T}^\infty
 \varphi(u) 
du \leq \varepsilon / 2$. Plugging this into \eqref{eq:term1supxint}, we get that for any $T$ large enough,  
$II^T \leq \varepsilon$. We can take $\varepsilon$ arbitrary small so that $II^T \to 0$. It remains to deal with $III^T$. We first define the uniform continuity modulus of $Z$ on $[0,M]$ by
\begin{align*}
\omega(u) = \sup_{t\leq M-u} |Z_t - Z_{t-u}|.
\end{align*}
It is well defined and almost surely continuous in $u$ since $Z$ is continuous. Moreover, $\omega$ is bounded and $\omega(0) = 0$. Therefore 
\begin{align*}
III^T = \sup_{t\leq M}
\Bigg |
\int_{0}^{t}
 T\varphi(uT) 
  \big(Z_t
  -
   Z_{t - u}\big)
 \, du
 \Bigg|
 \leq
 \int_{0}^{\infty}
 T\varphi(uT) 
  \omega(u)
 \, du
=
 \int_{0}^{\infty}
 \varphi(u) 
  \omega(u/T)
 \, du.
\end{align*}
Since $\omega$ is bounded, $\omega(u/T) \to 0$ as $T \to \infty$ and since $\varphi$ is integrable, we can use Lebesgue convergence theorem and we get $\int_{0}^{\infty}
 \varphi(u) 
  \omega(u/T)
 \, du \to 0$ as $T \to \infty$. Thus
\begin{align}
\label{eq:proof:limsup2}
\sup_{t\leq M}
\Bigg |
\int_{0}^{t}
 T\varphi(uT) 
  \big(Z_t
  -
   Z_{t - u}\big)
 \, du
 \Bigg|
\to 0.
\end{align}
Therefore, $\sup_{t \leq M}
|
\int_{0}^\infty
 \varphi(u) 
  \big(Z^T_t
  -
   Z^T_{(t - u/T)_+}\big)
 \, du
| 
\to 0$, which concludes the proof of Lemma~\ref{lemma:convol:gamma}.
\end{proof}

\subsection{Uniqueness of $r^*$}
\label{sec:proof:uniqueness}

Since $\MINL$ is increasing and non-negative, we first note that any solution $r$ of \eqref{eq:def:r} must be positive and therefore
\begin{align*}
\int_0^t \lambda^{-1} (t-s)^{-\alpha}(\gamma f(s) - r(s)) \, ds
\leq 
\frac{ \gamma \norm{f}_{\infty}}{\lambda(1-\alpha)}
t^{1-\alpha}
\;\; \text{ and }
r(t)
\leq 
\MINL \Big(\frac{\gamma \norm{f}_{\infty}}{\lambda(1-\alpha)}
t^{1-\alpha}\Big).
\end{align*}

Consider now $r_1$ and $r_2$ two solutions of \eqref{eq:def:r} on $[0,M]$. Since $\MINL(x) = 0$ for $x \leq 0$, we know that $\MINL$ is Lipschitz on $(-\infty, K \gamma (1-\alpha)^{-1} \norm{f}_{\infty}
M^{1-\alpha}]$. We write $|\MINL|_{lip}$ for its Lipschitz constant on this interval. We have
\begin{align}
\nonumber
|r_1(t) - r_2(t)| 
&\leq |\MINL|_{lip}
\Big| 
\int_0^t \lambda^{-1}(t-s)^{-\alpha} (\gamma f(s) - r_1(s))\, ds
-
\int_0^t \lambda^{-1}(t-s)^{-\alpha} (\gamma f(s) - r_2(s))\, ds
\Big|
\\
\label{eq:diff:r1r2}
&\leq |\MINL|_{lip}
\int_0^t \lambda^{-1}(t-s)^{-\alpha} | 
r_1(s) - r_2(s) |
\, ds
.
\end{align}
We now use a Gronwall type inequality for convolution equations with singular kernels. This result is due to  \cite{henry2006geometric}.
\begin{lemma}[Lemma 7.1.1 in \cite{henry2006geometric}]
Let $a$ be a nonnegative and locally integrable function on $[0,M]$ and let $b \geq 0$.  Suppose $u$ is a nonnegative locally integrable function on $[0,M]$ satisfying 
\begin{align*}
u(t) \leq a(t) + \int_0^t b(t-s)^{-\alpha} u(s) \, ds.
\end{align*}
Then 
\begin{align}
\label{eq:bound:u:gronwall}
u(t) \leq a(t) + \Big(\frac{b}{\Gamma(1-\alpha)}\Big)^{1/(1-\alpha)} \int_0^t G_{\alpha}'\Big(\Big(\frac{b}{\Gamma(1-\alpha)}\Big)^{1/(1-\alpha)} (t-s)\Big) a(s) \, ds
\end{align}
where $\delta = K\alpha^{-1}\Gamma(1-\alpha)$, $\Gamma(z) = \int_0^\infty t^{z-1}e^{-t}\,dt$ being the usual Gamma function and $G_{\alpha}'$ the derivative of $G_{\alpha}$ defined by $G_{\alpha}(z) = \sum_{n=0}^\infty z^{n(1-\alpha)}/\Gamma(n(1-\alpha)+1)$.
\end{lemma}
In particular, when $a = 0$, Equation \eqref{eq:bound:u:gronwall} means that $u(t) = 0$ for any $t \geq 0$. Therefore Equation \ref{eq:diff:r1r2} implies that $|r_1(t) - r_2(t)| = 0$ for all $t$ and thus $r_1 = r_2$.

\subsection{Study of $\overline{N}^{o,T}_t$}

We first prove that $\overline{N}^{o,T}$ is tight for the Skorokhod topology on $[0,M]$. By Theorem VI-4.1 in \cite{jacod1987limit}, it suffices to prove the following:
\begin{enumerate}[label={(\roman*)}]
\item $(\overline{N}^{o,T}_0)$ is tight,
\item $\lim\limits_{t\to 0} \limsup\limits_{T \to \infty} \mathbb{P}[|\overline{N}^{o,T}_t - \overline{N}^{o,T}_0| > \varepsilon] = 0$ for all $\varepsilon > 0$,
\item There exists $C > 0$, $\delta > 0$ and $\beta > 1$ such that for all $\lambda > 0$ and all $r \leq s \leq t$, we have
\begin{equation*}
\mathbb{P}[|\overline{N}^{o,T}_t - \overline{N}^{o,T}_s| > \lambda, \, |\overline{N}^{o,T}_s - \overline{N}^{o,T}_r| > \lambda] \leq C \lambda^{-\delta} |t-r|^\beta.
\end{equation*}
\end{enumerate}

The first point is satisfied since $\overline{N}^{o,T}_0 = 0$. By Markov's inequality, the two other points are satisfied once we prove that
$\mathbb{E}[\overline{N}^{o,T}_t] \to 0$ when $t \to 0$ and that for all $r \leq s \leq t$, we have
\begin{equation*}
\mathbb{E}\Big [\big|\overline{N}^{o,T}_t - \overline{N}^{o,T}_s\big| \big|\overline{N}^{o,T}_s - \overline{N}^{o,T}_r\big | \Big ] \leq C |t-r|^{-\beta}.
\end{equation*}

By definition, $N^{o,T}$ is an inhomogeneous Poisson process with intensity $\gamma \beta^T f(T^{-1} \cdot)$. Therefore, we have
\begin{equation*}
\mathbb{E}[\overline{N}^{o,T}_t] = (T\beta^T)^{-1} \int_0^{tT} \gamma \beta^T f(s/T) \,ds = \gamma \int_0^t f(s) \,ds.
\end{equation*}
Since $f$ in integrable, we have $\gamma \int_0^t f(s) \,ds \to 0$ as $t\to 0$ and therefore $\mathbb{E}[\overline{N}^{o,T}_t] \to 0$. Moreover, when $r \leq s \leq t$, we know that $\overline{N}^{o,T}_t - \overline{N}^{o,T}_s$ and $\overline{N}^{o,T}_s - \overline{N}^{o,T}_r$ are non-negative and independent and therefore
\begin{equation*}
\mathbb{E}\Big [\big|\overline{N}^{o,T}_t - \overline{N}^{o,T}_s\big| \big|\overline{N}^{o,T}_s - \overline{N}^{o,T}_r\big | \Big ] 
=
\mathbb{E}\big [\overline{N}^{o,T}_t - \overline{N}^{o,T}_s \big ] 
\mathbb{E}\big [\overline{N}^{o,T}_s - \overline{N}^{o,T}_r \big ] 
=
\gamma^2 \int_s^t f(u) \,du
\int_r^s f(u) \,du.
\end{equation*}
Using that $f$ is bounded, we have
\begin{equation*}
\mathbb{E}\Big [\big|\overline{N}^{o,T}_t - \overline{N}^{o,T}_s\big| \big|\overline{N}^{o,T}_s - \overline{N}^{o,T}_r\big | \Big ] 
\leq C (t-s)(s-r) \leq \frac{C}{4} (t-r)^2
\end{equation*}
%
%
which proves that $\overline{N}^{o,T}$ is tight for the Skorokhod topology on $[0,M]$. 

Moreover, the jumps of $\overline{N}^{o,T}$ are uniformly bounded by $(T\beta^T)^{-1} \to 0$ so by Theorem VI.3.26 of \cite{jacod1987limit}, $\overline{N}^{o,T}$ is even $C$-tight (meaning that the limit of a subsequence of $\overline{N}^{o,T}$ is continuous). Note then that $\overline{M}^{o,T}$ is a martingale and $\langle \overline{M}^{o,T} \rangle = \overline{N}^{o,T}$. By Theorems VI.3.26 and VI-4.13 in \cite{jacod1987limit}, we deduce that the sequence $\overline{M}^{o,T}$ is also $C$-tight. It remains to identify the limit of $\overline{N}^{o,T}$ to conclude Step 1. Let $(X^o, Z^o)$ be the limit of a subsequence of $(\overline{N}^{o,T}, \overline{M}^{o,T})$. By definition, we have
\begin{align*}
\overline{N}^{o,T}_t = (T\beta^T)^{-1/2} \overline{M}^{o,T}_t + \gamma \int_0^t f(s) \, ds.
\end{align*}
Since $(T\beta^T)^{-1/2} \overline{M}^{o,T}_t \to 0$, we must have
\begin{align*}
X^o_t = \gamma \int_0^t f(s) \, ds.
\end{align*}

\subsection{Study of $\widetilde{\lambda}^{o,T}_t$}
\label{sec:step:widetilde:lambda:oT}

In this section, we aim at proving that 
\begin{align*}
\sup_{t \in [0,T]}
\Big | \widetilde{\lambda}^{o,T}_t - \gamma \lambda^{-1} \int_0^t (t-s)^{-\alpha} f(s) \, ds \Big| \to 0
\end{align*}
in distribution. Recall that $\widetilde{\lambda}^{o,T}$ is defined in Equation \eqref{eq:def:lambdatilde} by
\begin{align*}
\widetilde{\lambda}^{o,T}_t 
&=
\int_0^t 
\overline{\zeta}^T(t-s) \, d \overline{N}^{o,T}_s
=
(T\beta^T)^{-1}
\int_0^{Tt} 
{\Gamma}^T(Tt-s) \, d N^{o,T}_s.
\end{align*}
By definition of $\overline{M}^{o,T}$, we have
\begin{align*}
\widetilde{\lambda}^{o,T}_t 
&=
(T\beta^T)^{-1/2}
\int_0^{t} 
\overline{\zeta}^T(t-s) \, d \overline{M}^{o,T}_s+ 
\gamma
\int_0^{t} 
\overline{\zeta}^T(t-s)  f(s) \, d s.
\end{align*}
We first focus on the integral with respect to $\overline{M}^{o,T}$. From the previous step, $\overline{M}^{o,T} \to Z^o$ uniformly on $[0,M]$. Moreover $Z^o$ is continuous and $Z^o_0 = \overline{M}^{o,T}_0 = 0$. Thus using Lemma~\ref{lemma:convol:gamma}, we get that 
\begin{align*}
(T\beta^T)^{-1/2}
\int_0^{t} 
\overline{\zeta}^T(t-s) \, d \overline{M}^{o,T}_s\to 0
\end{align*}
uniformly for $t \in [0,M]$. We are now left with $\gamma
\int_0^{t} 
\overline{\zeta}^T(t-s)  f(s) \, d s$. Using Lemma~\ref{lemma:convergence:gammameasure}, we have
\begin{align}
\label{eq:secondconvergencepointwise}
\int_0^{t} 
\overline{\zeta}^T(t-s)  f(s) \, d s 
\to 
\lambda^{-1}
\int_0^{t} 
(t-s)^{-\alpha}  f(s) \, d s
\end{align}
for all $0 \leq t \leq M$ pointwise. We will prove that this convergence is uniform in $t$ by showing that the sequence of functions $t \mapsto \int_0^{t} 
\overline{\zeta}^T(t-s)  f(s) \, d s$ lies in a compact of $C([0,M])$. By Arzelà–Ascoli theorem, we just need to show that these functions are uniformly equicontinuous, meaning that for any $\varepsilon > 0$, there exists $\delta > 0$ (independent of $T$) such that 
\begin{align}
\label{eq:goal:zeta}
\Big|
&\int_0^{t} 
\overline{\zeta}^T(s)  f(t-s) \, d s 
-
\int_0^{t+h} 
\overline{\zeta}^T(s)  f(t+h-s) \, d s 
\Big|
\leq \varepsilon
\end{align} 
for any $0 < h \leq \delta$ and  $0 \leq t \leq M-h$. Indeed, for any $0 \leq t \leq t+h \leq M$, we have
\begin{align*}
\Big|
&\int_0^{t} 
\overline{\zeta}^T(s)  f(t-s) \, d s 
-
\int_0^{t+h} 
\overline{\zeta}^T(s)  f(t+h-s) \, d s 
\Big|
\\
&\leq
\Big|
\int_0^{t} 
(\overline{\zeta}^T(t-s) - \overline{\zeta}^T(t+h-s)) f(s) \, d s 
-
\int_t^{t+h} 
\overline{\zeta}^T(t+h-s)
f(s) \, d s 
\Big|
\\
&\leq
\norm{f}_{\infty}
\int_0^{t} 
\Big|\overline{\zeta}^T(t-s) - \overline{\zeta}^T(t+h-s)\Big|  \, d s 
+
\norm{f}_{\infty}
\int_t^{t+h} 
\overline{\zeta}^T(t+h-s)
\, d s.
\end{align*}
Moreover, $\overline{\zeta}^T$ is decreasing so $\overline{\zeta}^T(t-s) \geq \overline{\zeta}^T(t+h-s)$. Therefore we get
\begin{align*}
\Big|
&\int_0^{t} 
\overline{\zeta}^T(s)  f(t-s) \, d s 
-
\int_0^{t+h} 
\overline{\zeta}^T(s)  f(t+h-s) \, d s 
\Big|
\\
&\leq
\norm{f}_{\infty}
\int_0^{t} 
\overline{\zeta}^T(t-s) - \overline{\zeta}^T(t+h-s)  \, d s 
+
\norm{f}_{\infty}
\int_t^{t+h} 
\overline{\zeta}^T(t+h-s)
\, d s
\\
&\leq
\norm{f}_{\infty}
\Big[\int_0^{t} 
\overline{\zeta}^T(s)  \, d s 
-
\int_h^{t+h} 
\overline{\zeta}^T(s)  \, d s 
\Big]
+
\norm{f}_{\infty}
\int_0^{h} 
\overline{\zeta}^T(s)
\, d s
\\
&\leq
2
\norm{f}_{\infty}
\int_0^{h} 
\overline{\zeta}^T(s)
\, d s.
\end{align*}
By Lemma~\ref{lemma:convergence:gammameasure}, we know that $\int_0^{h} \overline{\zeta}^T(s)  \, d s \to \lambda^{-1} (1-\alpha)^{-1} h^{1-\alpha}$ as $T\to \infty$. But $h \mapsto \int_0^{h} \overline{\zeta}^T(s)  \, d s$ is a sequence of increasing functions so the convergence $\int_0^{h} \overline{\zeta}^T(s)  \, d s \to \lambda^{-1} (1-\alpha)^{-1} h^{1-\alpha}$ is uniform in $0 \leq h \leq M$. Therefore we can apply Arzelà–Ascoli theorem to $h \mapsto \int_0^{h} \overline{\zeta}^T(s)  \, d s$ so there exists $\delta > 0$ independent of $T$ such that for any $0 < h \leq \delta$ and  $0 \leq t' \leq M-h$. Consequently, we have
\begin{align*}
\Big|
&\int_0^{t'} 
\overline{\zeta}^T(s)  \, d s
-
\int_0^{t'+h} 
\overline{\zeta}^T(s)  \, d s
\Big|
\leq \frac{\varepsilon}{2\norm{f}_\infty}.
\end{align*} 
We get Equation \eqref{eq:goal:zeta} by taking $t' = 0$. Therefore the convergence of Equation \eqref{eq:secondconvergencepointwise} is uniform.

\subsection{Study of $\overline{N}^{m,T}_t$}

We now focus on the sequences $\overline{N}^{m,T}$ and $\overline{M}^{m,T}$. We cannot use the same tools as for $\overline{N}^{o,T}$ and $\overline{M}^{o,T}$ because we cannot compute explicitly the expectation $\mathbb{E}[\overline{N}^{m,T}_t]$ for $0 \leq t \leq M$. Instead, we prove that $\overline{N}^{m,T}$ is tight for the Skorokhod topology on $[0,M]$ using Aldous' criteria stated in  Theorem VI.3.26 in \cite{jacod1987limit}. More precisely, we need to prove the following two conditions
\begin{enumerate}[label={(\roman*)}]
\item For all $\varepsilon > 0$, there exists $T_0 > 0$ and $K > 0$ such that for all $T > T_0$, 
\begin{equation*}
\mathbb{P}
\Big[
\sup_{t\leq M}
\big|
\overline{N}^{m,T}_t
\big|
> K
\Big] \leq \varepsilon.
\end{equation*}
\item We write $\mathcal{T}_M^T$ for the set of $(\mathcal{F}_{tT})_t$-stopping times $S$ bounded by $M$. For all $\varepsilon > 0$, we have
\begin{equation}
\label{eq:stoppingtimecond}
\lim\limits_{\delta \to 0}
\limsup\limits_{T \to \infty}
\sup_{R,S \in \mathcal{T}_M,\,  R \leq S \leq R + \delta}
\mathbb{P}
\Big [
\big |
\overline{N}^{m,T}_S - \overline{N}^{m,T}_R
\big |
\geq \varepsilon
\Big]
=
0.
\end{equation}
\end{enumerate}

Note that (i) is a consequence of the bound on $\mathbb{E}[\overline{N}^{o,T}_M]$, of Markov's inequality, and of the following comparison lemma.
\begin{lemma}
\label{lemma:comparison}
Almost surely, $0 \leq N^{m,T}_t \leq N^{o,T}_t$ for any $t \geq 0$.
\end{lemma}
\begin{proof}
Let $\tau = \inf \{ t \geq 0: N^{m,T}_t > N^{o,T}_t\}$, with the convention $\inf \emptyset = \infty$. We aim at proving that $\tau = \infty$ almost surely.\\

Up to a probability space enlargement, we can always suppose by Theorem \ref{thm:whole_model} that for all $t \geq 0$
\begin{align*}
N^{m,T}_t = \int_0^t \int_0^\infty \1_{z \leq \lambda^{m,T}_{s}} \, \pi^{m,T}(ds\,dz)
\end{align*}
for some random Poisson point measure $\pi^{m,T}$. By Theorem \ref{thm:whole_model}, it is also clear that the trajectories of $N^{m,T}$ and $N^{o,T}$ are almost surely right continuous with left limits and the jump size is always $1$ and that almost surely, $N^{m,T}$ and $N^{o,T}$ do not jump at time $0$. This implies in particular that $\tau > 0$ almost surely.\\

Suppose that $\omega$ is such that $0 < \tau(\omega) < \infty$ and  $N^{m,T}(\omega)$ and $N^{o,T}(\omega)$ are almost surely right continuous with left limits and the jump size is always $1$. We want to prove that $\omega$ is necessarily in a subset of $\Omega$ of probability $0$ and we don't write $\omega$ in the following for readability. For all $t < \tau$, we have $N^{m,T}_t \leq N^{o,T}_t$ by definition of $\tau$. Since the jump size of  $N^{m,T}$ and $N^{o,T}$ is $1$, we have $N^{m,T}_\tau = N^{o,T}_\tau + 1$. Let $N = N^{o,T}_\tau$. We write $T^{o,T}_1, \dots, T^{o,T}_N$ for the jump times of $ N^{o,T}$ and $T^{m,T}_1, \dots, T^{m,T}_{N+1}$ for the jump times of $N^{m,T}$ on $[0, \tau]$. Since $N^{m,T}$ jumps at time $\tau$, this means that $\pi^{m,T}(\{\tau\} \times [0, \lambda^{m,T}_{\tau}]) > 0$. But by Equation \eqref{eq:signal}, we have
\begin{align*}
\lambda^{m,T}_{\tau} 
&= 
\MINL \big( 
\kappa 
\int_0^{\tau-} 
\zeta^T(\tau-s) \, d (N^{o,T}_s - N^{m,T}_s)
 \big) 
= 
\MINL^T \Big( 
\kappa 
\sum_{k=1}^N
\zeta^T(\tau-T^{o,T}_k) - \zeta^T(\tau-T^{m,T}_k)  
\Big).
\end{align*}
By definition of $\tau$, we have $T^{o,T}_k \leq T^{m,T}_k$. Therefore, using that $\zeta^T$ is decreasing and that $\MINL^T(x) = 0$ for $x \leq 0$, we have $\lambda^{m,T}_{\tau} = 0$. We deduce that $\pi^{m,T}(\{(\tau,0)\} ) > 0$.\\

Therefore we have
\begin{align*}
\mathbb{P}[\{ \omega: \, 0 < \tau(\omega) < \infty\} ]
&\leq \mathbb{P}\Big[\Big\{ \omega: \, \big (\pi^{m,T}(\omega)\big )\big (\{(\tau(\omega),0)\} \big ) > 0 \Big\} \Big]
\\
&\leq \mathbb{P} [\pi^{m,T}([0,\infty) \times \{0\}]
\\
&= 0.
\end{align*}
Thus $\tau = \infty$ almost surely, which proves Lemma~\ref{lemma:comparison}.
\end{proof}

We now focus on proving \eqref{eq:stoppingtimecond}.
Note first that by the previous step, $\widetilde{\lambda}^{o,T}_t$ converges uniformly on $[0,M]$ and is therefore $C$-tight. By Theorem VI.3.26 in \cite{jacod1987limit}, we deduce that for all $\iota > 0$, there exists $V > 0$ such that 
\begin{align*}
\mathbb{P} \Big[ \sup_{t\leq M} \,\widetilde{\lambda}^{o,T}_t \geq V \Big] \leq \iota, \;\; \text{ and thus }\;\;
\mathbb{P} \Big[ \sup_{t\leq M} \, \overline{\lambda}^{m,T}_t \geq \MINL(V) \Big] \leq \iota.
\end{align*}
Thus using the same notations as in \eqref{eq:stoppingtimecond}, we have
\begin{equation}
\label{eq:decompositionterm2}
\begin{split}
\mathbb{P}
\Big [
\big |
\overline{N}^{m,T}_S - \overline{N}^{m,T}_R
\big |
\geq \varepsilon
\Big]
&\leq
\mathbb{P}
\Big [
\big |
\overline{M}^{m,T}_S - \overline{M}^{m,T}_R
\big |
\geq (T\beta^T)^{1/2} \varepsilon, \; 
\sup_{t\leq M} \, \overline{\lambda}^{m,T}_t \leq \MINL(V)
\Big]
\\&\;\;\;\;+
\mathbb{P}
\Big [
\big |
\int_R^S \overline{\lambda}^{m,T}(s) \, ds
\big |
\varepsilon, \; 
\sup_{t\leq M} \, \overline{\lambda}^{m,T}_t \leq \MINL(V)
\Big]
+ \iota.
\end{split}
\end{equation}
Then, note that up to a probability space enlargement, we can always suppose by Theorem \ref{thm:whole_model} that for all $t \geq 0$
\begin{align*}
N^{m,T}_t = \int_0^t \int_0^\infty \1_{z \leq \lambda^{m,T}_{s}} \, \pi^{m,T}(ds\,dz)
\end{align*}
for some random Poisson point measure $\pi^{m,T}$ and therefore we can write
\begin{equation*}
\overline{M}^{m,T}_t = (T\beta^T)^{-1/2} \int_0^{tT} \int_0^\infty \1_{z \leq \lambda^{m,T}_{s}} \, \widetilde{\pi}^{m,T}(ds\,dz)
\end{equation*}
where
$\widetilde{\pi}^{m,T}(ds\,dz) = {\pi}^{m,T}(ds\,dz) - ds\,dz$. Therefore
\begin{equation*}
\overline{M}^{m,T}_S - \overline{M}^{m,T}_R = (T\beta^T)^{-1/2} \int_{RT}^{ST} \int_0^\infty \1_{z \leq \lambda^{m,T}_{s}} \, \widetilde{\pi}^{m,T}(ds\,dz)
\end{equation*}
and on the event $\{ \sup_{t\leq M} \, \overline{\lambda}^{m,T}_t \leq \MINL(V) \}$, we can replace $\lambda^{m,T}_{s}$ by $\lambda^{m,T}_{s} \wedge (\beta^T \MINL(V))$. Using Markov's inequality, we obtain
\begin{align*}
\mathbb{P}
\Big [
\big |
\overline{M}^{m,T}_S - \overline{M}^{m,T}_R
\big |
\geq (T\beta^T)^{1/2} &\varepsilon, \; 
\sup_{t\leq M} \, \overline{\lambda}^{m,T}_t \leq \MINL(V)
\Big]
\\
&\leq
\mathbb{P}
\Big [
\Big| \int_{RT}^{ST} \int_0^\infty \1_{z \leq \lambda^{m,T}_{s} \wedge (\beta^T \MINL(V))} \, \widetilde{\pi}^{m,T}(ds\,dz)
\Big| 
\geq
T\beta^T \varepsilon
\Big]
\\
&\leq
(T\beta^T \varepsilon)^{-2}
\mathbb{E}
\Big [
\Big| \int_{RT}^{ST} \int_0^\infty \1_{z \leq \lambda^{m,T}_{s} \wedge (\beta^T \MINL(V))} \, \widetilde{\pi}^{m,T}(ds\,dz)
\Big|^2
\Big].
\end{align*}
Using also that the process
$(\int_{0}^{t} \int_0^\infty \1_{z \leq \lambda^{m,T}_{s} \wedge (\beta^T \MINL(V))} \, \widetilde{\pi}^{m,T}(ds\,dz) )_t$ 
is a martingale whose quadratic variation is given by $
( \int_{0}^{t} \int_0^\infty \1_{z \leq \lambda^{m,T}_{s} \wedge (\beta^T \MINL(V))} \, \pi^{m,T}(ds\,dz) )_t$ 
and using that $RT$ and $ST$ are $(\mathcal{F}_t)_t$-stopping times, we get
\begin{align*}
\mathbb{P}
\Big [
\big |
\overline{M}^{m,T}_S - \overline{M}^{m,T}_R
\big |
\geq (T\beta^T)^{1/2} &\varepsilon, \; 
\sup_{t\leq M} \, \overline{\lambda}^{m,T}_t \leq \MINL(V)
\Big]
\\
&\leq
(T\beta^T \varepsilon)^{-2}
\mathbb{E}
\Big [
\int_{RT}^{ST} \int_0^\infty \1_{z \leq \lambda^{m,T}_{s} \wedge (\beta^T \MINL(V))} \, {\pi}^{m,T}(ds\,dz)
\Big]
\\
&\leq
(T\beta^T \varepsilon)^{-2}
\mathbb{E}
\Big [
\int_{RT}^{ST} \int_0^\infty \1_{z \leq \lambda^{m,T}_{s} \wedge (\beta^T \MINL(V))} \, ds\,dz
\Big]
\\
&\leq
(T\beta^T \varepsilon)^{-2}
\mathbb{E}
\Big [
(RT-ST) \beta^T \MINL(V)
\Big].
\end{align*}
Moreover, we know that $S-R \leq \delta$ by definition and therefore
\begin{align*}
\mathbb{P}
\Big [
\big |
\overline{M}^{m,T}_S - \overline{M}^{m,T}_R
\big |
\geq (T\beta^T)^{1/2} &\varepsilon, \; 
\sup_{t\leq M} \, \overline{\lambda}^{m,T}_t \leq \MINL(V)
\Big]
\leq
(T\beta^T \varepsilon)^{-1} \delta
\MINL(V)
\to 0.
\end{align*}
We now turn to the second probability in \eqref{eq:decompositionterm2}. We have
\begin{align*}
\mathbb{P}
\Big [
\big |
\int_R^S \overline{\lambda}^{m,T}(s) \, ds
\big |
\varepsilon, \; 
\sup_{t\leq M} \, \overline{\lambda}^{m,T}_t \leq \MINL(V)
\Big]
\leq
\mathbb{P}
\Big [
(R-S) \MINL(V)
\geq
\varepsilon
\Big] = 0
\end{align*}
whenever $\delta \MINL(V) > \varepsilon$ by definition of $R$ and $S$.
Therefore, we obtain 

\begin{equation*}
\limsup\limits_{\delta \to 0}
\limsup\limits_{T \to \infty}
\sup_{R,S \in \mathcal{T}_M,\,  R \leq S \leq R + \delta}
\mathbb{P}
\Big [
\big |
\overline{N}^{m,T}_S - \overline{N}^{m,T}_R
\big |
\geq \varepsilon
\Big]
\leq
\iota
\end{equation*}
for all $\iota > 0$ which proves \eqref{eq:stoppingtimecond}. We conclude that $\overline{N}^{m,T}$ is tight for the Skorokhod topology on $[0,M]$. Then, using the same arguments as for $\overline{N}^{o,T}$ and $\overline{M}^{o,T}$, we see that $\overline{N}^{m,T}$ and $\overline{M}^{o,T}$ are both C-tight. Moreover, if $( X^m, Z^m)$ is a limit of a  subsequence of $( \overline{N}^{m,T}, \overline{M}^{m,T})$, $Z^m$ is a martingale with quadratic variation $X^m$.

\subsection{Study of $\widetilde{\lambda}^{m,T}_t$}
\label{sec:step:widetilde:lambda:mT}

By definition of $\overline{M}^{m,T}$, we have
\begin{align*}
\widetilde{\lambda}^{m,T}_t 
&=
(T\beta^T)^{-1/2}
\int_0^{t} 
\overline{\zeta}^T(t-s) \, d \overline{M}^{m,T}_s + 
\int_0^{t} 
\overline{\zeta}^T(t-s)  \overline{\lambda}^{m,T}_s 
 \, d s.
\end{align*}
Using Step 1, we easily check that $\overline{M}^{m,T}$ satisfies the hypotheses of Lemma~\ref{lemma:convol:gamma} so that
\begin{align}
\label{eq:proof:convMmT}
(T\beta^T)^{-1/2}
\int_0^{t} 
\overline{\zeta}^T(t-s) \, d \overline{M}^{m,T}_s \to 0
\end{align}
uniformly on $[0,M]$.
We now prove that $\int_0^{t} 
\overline{\zeta}^T(t-s)  \overline{\lambda}^{m,T}_s 
 \, d s$ is $C$-tight using Theorem VI.3.26 in \cite{jacod1987limit}. More precisely, we need to prove the following two conditions
\begin{enumerate}[label={(\roman*)}]
\item For all $\varepsilon > 0$, there exists $T_0 > 0$ and $K > 0$ such that for all $T > T_0$, 
\begin{equation*}
\mathbb{P}
\Big[
\sup_{t\leq M}
\Big|
\int_0^{t} 
\overline{\zeta}^T(t-s)  \overline{\lambda}^{m,T}_s 
 \, d s
\Big|
> K
\Big] \leq \varepsilon.
\end{equation*}
\item For all $\varepsilon > 0$ and $\eta > 0$, there exists $T_0 > 0$ and $\delta > 0$ such that for all $T > T_0$, 
\begin{equation}
\label{eq:goal:boundinc}
\mathbb{P}
\Big[
\sup_{t\leq M} \sup_{h \leq \delta \wedge (M-t)}
\Big|
\int_0^{t+h} 
\overline{\zeta}^T(t+h-s)  \overline{\lambda}^{m,T}_s 
 \, d s
-
\int_0^{t} 
\overline{\zeta}^T(t-s)  \overline{\lambda}^{m,T}_s 
 \, d s
\Big|
> \eta
\Big] \leq \varepsilon.
\end{equation}
\end{enumerate}

We first show that (i) holds. We fix $\varepsilon > 0$.
Note that since $\MINL$ is increasing and $\widetilde{\lambda}^{m,T}_t$ is positive, we have
\begin{align*}
\overline{\lambda}^{m,T}_t = \MINL(\widetilde{\lambda}^{o,T}_t - \widetilde{\lambda}^{m,T}_t) \leq \MINL(\widetilde{\lambda}^{o,T}_t).
\end{align*} 
By the previous step, $\widetilde{\lambda}^{o,T}_t$ converges uniformly on $[0,M]$ and is therefore $C$-tight. By Theorem VI.3.26 in \cite{jacod1987limit}, we deduce that there exists $V > 0$ such that 
\begin{align*}
\mathbb{P} \Big[ \sup_{t\leq M} \,\widetilde{\lambda}^{o,T}_t \geq V \Big] \leq \varepsilon, \;\; \text{ and thus }\;\;
\mathbb{P} \Big[ \sup_{t\leq M} \, \overline{\lambda}^{m,T}_t \geq \MINL(V) \Big] \leq \varepsilon.
\end{align*}
Moreover, $\int_0^{M} 
\overline{\zeta}^T(s) 
 \, d s$ converges as $T \to \infty$ by Lemma~\ref{lemma:convergence:gammameasure} so it is bounded by a constant $L > 0$. Then we get
\begin{align*}
\sup_{t\leq M}
\Big|
\int_0^{t} 
\overline{\zeta}^T(t-s)  \overline{\lambda}^{m,T}_s 
 \, d s
\Big|
\leq
L
\sup_{t\leq M}
\overline{\lambda}^{m,T}_t 
\end{align*}
so that we obtain
\begin{align*}
\mathbb{P} \Big[ \sup_{t\leq M}
\Big|
\int_0^{t} 
\overline{\zeta}^T(t-s)  \overline{\lambda}^{m,T}_s 
 \, d s
\Big| \geq \MINL(V) L \Big] \leq \varepsilon
\end{align*}
which proves that (i) holds. It remains to prove \eqref{eq:goal:boundinc}. Proceeding as in Section \ref{sec:step:widetilde:lambda:oT}, we show that for all $t \geq 0$ and $h \geq 0$;
\begin{align*}
\Big|
&\int_0^{t} 
\overline{\zeta}^T(t-s)  \overline{\lambda}^{m,T}_s \, d s 
-
\int_0^{t+h} 
\overline{\zeta}^T(t+h-s)  \overline{\lambda}^{m,T}_s \, d s 
\Big|
\leq 
2
\int_0^{h} 
\overline{\zeta}^T(s) \, d s \sup_{u \leq M} \, \overline{\lambda}^{m,T}_u.
\end{align*}
Recall also from Section \ref{sec:step:widetilde:lambda:oT} that the sequence of functions $h \mapsto \int_0^{h} \overline{\zeta}^T(s)  \, d s$ is equicontinuous so there exists $\delta > 0$ independent of $T$ such that for any $0 < h \leq \delta$ and  $0 \leq t \leq M-h$, we have
\begin{align*}
\Big|
&\int_0^{t} 
\overline{\zeta}^T(s)  \, d s
-
\int_0^{t+h} 
\overline{\zeta}^T(s)  \, d s
\Big|
\leq \frac{\eta}{2\MINL(V)}.
\end{align*}
In particular, 
$|
\int_0^{h} 
\overline{\zeta}^T(s)  \, d s
|
\leq \frac{\eta}{2\MINL(V)}$ so that we get
\begin{align*}
\mathbb{P}
\Big[
\sup_{t\leq M} \sup_{h \leq \delta \wedge (M-t)}
\Big|
&\int_0^{t+h} 
\overline{\zeta}^T(t+h-s)  \overline{\lambda}^{m,T}_s 
 \, d s
-
\int_0^{t} 
\overline{\zeta}^T(t-s)  \overline{\lambda}^{m,T}_s 
 \, d s
\Big|
> \eta
\Big]
\\
&\leq
\mathbb{P}
\Big[
2
\sup_{h\leq\delta}
\int_0^{h} 
\overline{\zeta}^T(s) \, d s \,
\sup_{u \leq M} \, \overline{\lambda}^{m,T}_u
> \eta
\Big]
\\
&\leq
\mathbb{P}
\Big[
\sup_{u \leq M} \, \overline{\lambda}^{m,T}_u
> \MINL(V)
\Big]
\leq \varepsilon.
\end{align*}
This proves Equation \eqref{eq:goal:boundinc} and therefore $\widetilde{\lambda}^{m,T}_t$ is C-tight by Theorem VI.3.26 in \cite{jacod1987limit}.

\subsection{Study of $\overline{\lambda}^{m,T}_t$}

We have seen in the previous steps that both $\widetilde{\lambda}^{o,T}$ and $\widetilde{\lambda}^{m,T}$ are $C$-tight.
By Corollary VI.3.33 in \cite{jacod1987limit}, $\widetilde{\lambda}^{o,T} - \widetilde{\lambda}^{m,T}$ is also $C$-tight. Since $\MINL$ is locally Lipschitz, we deduce using again Theorem VI.3.26 in \cite{jacod1987limit} that $\overline{\lambda}^{m,T}_t$ is $C$-tight. We write $r$ the (random) limit of any subsequence of $\overline{\lambda}^{m,T}$. We now prove that $r = r^*$. To do so, we 
show that it is solution of the integral equation \eqref{eq:def:r} and by Section \ref{sec:proof:uniqueness}, we know \eqref{eq:def:r} admits at most one solution so necessarily $r=r^*$.

Let $0 \leq t \leq M$. First recall that by definition of $\overline{M}^{m,T}$, we have
\begin{align*}
\widetilde{\lambda}^{m,T}_t 
&=
(T\beta^T)^{-1/2}
\int_0^{t} 
\overline{\zeta}^T(t-s) \, d \overline{M}^{m,T}_s + 
\int_0^{t} 
\overline{\zeta}^T(t-s)  \overline{\lambda}^{m,T}_s 
 \, d s
\\
&=
(T\beta^T)^{-1/2}
\int_0^{t} 
\overline{\zeta}^T(t-s) \, d \overline{M}^{m,T}_s 
+ 
\int_0^{t} 
\overline{\zeta}^T(t-s)  r(s)
 \, d s
+ 
\int_0^{t} 
\overline{\zeta}^T(t-s)  (\overline{\lambda}^{m,T}_s - r(s))
 \, d s.
\end{align*}
Then from Equation \eqref{eq:proof:convMmT} in Section \ref{sec:step:widetilde:lambda:mT}, we know that
\begin{align*}
\sup_{t \leq M} \Big|
(T\beta^T)^{-1/2}
\int_0^{t} 
\overline{\zeta}^T(t-s) \, d \overline{M}^{m,T}_s 
\Big|
\to 0,
\end{align*}
and by Lemma~\ref{lemma:convergence:gammameasure}, we have for all $t \geq 0$
\begin{align*}
\int_0^{t} 
\overline{\zeta}^T(t-s) r(s)
 \, d s
\to
\int_0^{t} 
\lambda^{-1} (t-s)^{-\alpha}  r(s)
 \, d s.
\end{align*}
Moreover, $\overline{\lambda}^{m,T} \to r$ uniformly on $[0,M]$ as $T \to \infty$ , so we also have
\begin{align}
\label{eq:convergence:lambdarint}
\Big|
\int_0^{t} 
\overline{\zeta}^T(t-s)  (\overline{\lambda}^{m,T}_s - r(s))
 \, d s
\Big|
\leq
\int_0^{t} 
\overline{\zeta}^T(s)
 \, d s
\,
\norm{\overline{\lambda}^{m,T} - r}_\infty
\to 0
\end{align}
as $T \to \infty$. Thus combining these three limits, we get that for all $t \geq 0$
\begin{align*}
\widetilde{\lambda}^{m,T}_t 
\to
\int_0^{t} 
\lambda^{-1} (t-s)^{-\alpha}  r(s)
 \, d s.
\end{align*}
Moreover, we already know by the results of Section \ref{sec:step:widetilde:lambda:oT} that $\widetilde{\lambda}^{o,T}_t \to \gamma \lambda^{-1} \int_0^t (t-s)^{-\alpha} f(s) \, ds$. Then recall that $\overline{\lambda}^{m,T}_t = \MINL(\widetilde{\lambda}^{o,T}_t - \widetilde{\lambda}^{m,T}_t)$. We get by continuity of $\MINL$
\begin{align*}
r(t) = \lim\limits_{T\to\infty} \overline{\lambda}^{m,T}_t = \lim\limits_{T\to\infty} \MINL(\widetilde{\lambda}^{o,T}_t - \widetilde{\lambda}^{m,T}_t)
=
\MINL\Big(
 \lambda^{-1} \int_0^t (t-s)^{-\alpha} (\gamma f(s) - r(s)) \, ds
\Big)
\end{align*}
so $r$ is solution of the integral equation \eqref{eq:def:r} and therefore $r=r^*$. \\

We are now ready to prove Equation \eqref{eq:limit:mi}. First note that from Definition \ref{def:marketimpact}, Equation \eqref{eq:scaling:mi} and the notations introduced in Section \ref{sec:outline} we have
\begin{align*}
\overline{MI}^T_t 
&=
\int_0^t 
1 + \overline{\zeta}^T(t-s) \, d (\overline{N}^{o,T}_s - \overline{N}^{m,T}_s)
\\
&=
\frac{1}{\sqrt{T\beta^T}}
\int_0^t 
(1 + \overline{\zeta}^T(t-s)) \, d (\overline{M}^{o,T}_s - \overline{M}^{m,T}_s)
+
\int_0^t 
(1 + \overline{\zeta}^T(t-s)) (\gamma f(s) - \overline{\lambda}^{m,T}_s) \, d s.
\end{align*}
The first integral $(T\beta^T)^{-1/2}
\int_0^t 
(1 + \overline{\zeta}^T(t-s)) \, d (\overline{M}^{o,T}_s - \overline{M}^{m,T}_s)$ converges to $0$ uniformly on $0 \leq t \leq M$ by Lemma~\ref{lemma:convol:gamma}. For the second integral, we repeat the proof of Section \ref{sec:step:widetilde:lambda:oT} to prove that 
\begin{align}
\label{eq:uniform:1}
\int_0^t 
(1 + \overline{\zeta}^T(t-s)) \gamma f(s) \, d s
\to 
\gamma 
\int_0^t 
(1 + K(t-s)^{-\alpha}) f(s) \, d s
\end{align}
uniformly on $0 \leq t \leq M$ and that of Section 
\ref{sec:step:widetilde:lambda:mT} to prove that 
$\int_0^t 
(1 + \overline{\zeta}^T(t-s)) \overline{\lambda}^{m,T}_s \, d s
$ is C-tight. Since we already know that $\overline{\lambda}^{m,T} \to r^*$ uniformly, using the same method as in Equation \eqref{eq:convergence:lambdarint}, we have
\begin{align}
\label{eq:uniform:2}
\int_0^t 
(1 + \overline{\zeta}^T(t-s)) \overline{\lambda}^{m,T}_s \, d s
\to
\int_0^t 
(1 + K(t-s)^{-\alpha}) r^*(s) \, d s,
\end{align} 
pointwise, for any $0 \leq t \leq M$. Eventually, since the sequence of process $\int_0^t 
(1 + \overline{\zeta}^T(t-s)) \overline{\lambda}^{m,T}_s \, d s$ is C-tight, this means that the last limit is uniform in $t$. Combining \eqref{eq:uniform:1} and \eqref{eq:uniform:2}, we get Equation \eqref{eq:limit:mi} which proves Theorem \ref{thm:market_impact}.

\section{Proof of the results of Section \ref{sec:macroscopic:limit}}
\label{sec:proof:macroscopic:limit}

\subsection{Notation}
\label{sec:macroscopic:notation}

Suppose that $\MINL$ is nondecreasing continuous and locally lipschitz. We define first 
\begin{align}
\label{eq:def:u}
u(\gamma, t) = \int_0^t \lambda^{-1} (t-s)^{-\alpha}(\gamma f(s) - r^*(\gamma, s)) \, ds
\end{align}
so that $r^*(\gamma, t) = \MINL(u(\gamma, t))$. By definition of $r^*$, we check that $u$ satisfies the non-linear Volterra integral equation
\begin{align}
\label{eq:equation:u}
u(\gamma, t) = 
\gamma
F_{\alpha, \lambda}(t) + 
\int_0^t 
K_\alpha(t-s) 
G_{\alpha, \lambda}(u(\gamma, s)) \, ds
\end{align}
where $K_\alpha(t) = \Gamma(1-\alpha)^{-1} t^{-\alpha}$ satisfies Condition \eqref{eq:condition:kernel}, $G_{\alpha, \lambda}(x) = - \lambda^{-1}\Gamma(1-\alpha) \MINL(x)$ and $F_{\alpha, \lambda}(t) = \int_0^t \lambda^{-1} (t-s)^{-\alpha} f(s)\, ds$. Moreover, $G_{\alpha, \lambda}$ and $F_{\alpha, \lambda}$ are continuous provided $f$ is locally bounded and $\MINL$ is continuous. Therefore $u$ is also continuous.

\subsection{Comparison theorem for non-linear Volterra equation}
\label{sec:non_linear:comparison}

In this section, we present some key results for the non-linear Volterra Equation \eqref{eq:equation:u}.

\begin{definition}
Let $k: [0,\infty) \to [0,\infty)$ locally integrable such that for any $t_0 > 0$, we have
\begin{equation}
\label{eq:condition:kernel}
\lim\limits_{t \to t_0} \int_0^{\infty} |k(t_0 - s) - k(t-s)| \, ds \to 0.
\end{equation}
Let $f: [0, \infty) \to \mathbb{R}$ and $g: \mathbb{R} \to \mathbb{R}$. We write $(E_{k,f,g})$ the Volterra integral Equation given by
\begin{equation*}
y(t) = f(t) + \int_0^t k(t-s) g(y(s)) \, ds
\end{equation*}
where $y: [0,\infty) \to \mathbb{R}$.
\end{definition}

Non-linear Volterra integral equations have been studied extensively and we refer to \cite{miller1971nonlinear} for a complete presentation (note that in \cite{miller1971nonlinear}, the setup is more general.) The main goal of this section is to derive comparison theorems between the solutions of several non-linear Volterra integral Equation. Such comparison results were already studied in \cite{miller1971nonlinear}. Before explaining the novelties in our context, let us recall the main results from \cite{miller1971nonlinear}. The comparison theorem is based on the concept of maximal solutions.

\begin{definition}
A solution $y^*$ of $(E_{k,f,g})$ is said to be maximal on a given compact subset of $[0, \infty)$ if, for any solution $y$ of $(E_{k,f,g})$ on this subset, we have $y \leq y^*$.
\end{definition}

\begin{theorem}[Chapter II. Theorem 6.1. in \cite{miller1971nonlinear}]
\label{thm:comparison:old}
For $i=1, 2$, we consider $F_i : [0, \infty) \to \mathbb{R}$ continuous and $G_i : [0, \infty) \to \mathbb{R}$ continuous and bounded.
Suppose that $|F_1(t)| \leq F_2(t)$ and that $|G_1(x)| \leq G_2(|x|)$ for all $t \geq 0$ and all $x \in \mathbb{R}$. Suppose that the maximal solution $z_2^*$ of $(E_{k,{F_2},{G_2}})$ exists and is defined on $[0, \beta]$ for some $\beta$. Then any solution $z_1^*$ of $(E_{k,{F_1},{G_1}})$ can be extended at least on $[0,\beta]$ and we have
\begin{align*}
|z_1(t)| \leq z_2^*(t) \, \text{ on } [0,\beta].
\end{align*}
\end{theorem}

One of the main issue for using of this result is the presence of maximal solutions instead of solutions. The assumptions on $F_1, F_2$ are also too weak for our applications and we need to generalize this result. Since any solution is maximal once we prove uniqueness of the solution, we start by proving uniqueness under some conditions of the solutions of $(E_{k,f,g})$.

\begin{proposition}
\label{propo:uniqueness}
Suppose that $f$ is a locally bounded function and that $g$ is locally Lipschitz continuous. Then there exists at most one locally bounded solution of $(E_{k,f,g})$ on any compact interval.
\end{proposition}

This result relates to Chapter II. Theorem 1.2. of \cite{miller1971nonlinear} where the uniqueness is proved only in \textit{some} neighbourhood of $0$, while we prove here uniqueness on \textit{any} neighbourhood of $0$.

\begin{proof}
Let $\beta > 0$. Suppose first that $g$ is bounded. Then any solution $y$ of $(E_{k,f,g})$ satisfies the inequality
\begin{align*}
|y(t)| \leq \sup_{[0,\beta]} |f| + \sup_{[0,\infty)} |g| \int_0^\beta k(s) \, ds
\end{align*}
for any $0 \leq t \leq \beta$. Therefore, $y$ is bounded on $[0, \beta]$ by a constant $C'$ that only depends on $k$, $f$, $g$ and $\beta$. Recall also that $g$ is Locally Lipschitz continuous and we write $|g|_{Lip, C'}$ for its Lipschitz continuity constant of $g$ on the interval $[0,C']$.

Let $y_1$ and $y_2$ be two solutions of Equation $(E_{k,f,g})$ on the compact interval $[0, \beta]$. We define
\begin{align*}
M = \sup \{ t\geq 0 : \, y_1(t) = y_2(t)\}.
\end{align*} 
Note that $M \geq 0$ is well defined since $y_1(0) = f(0) = y_2(0)$. We want to prove that necessarily, $M = \beta$. For any $M < t \leq \beta$, we have
\begin{align*}
| y_1(t) - y_2(t)| 
&\leq \int_\beta^t k(t-s) \, |g(y_1(s)) - g(y_2(s))| \, ds
\leq |g|_{Lip, C'} \int_\beta^t k(t-s) \, |y_1(s)-y_2(s)| \, ds.
\end{align*}
Therefore, if we write $\delta(t) = \sup_{s \leq t} | y_1(s) - y_2(s)|$, we have for any $M < t \leq \beta$
 \begin{align*}
\delta(t)
&\leq |g|_{Lip, C'} \delta(t) \int_0^{t-\beta} k(s) \, ds.
\end{align*}
Assume that $M < \beta$. Since $k$ is locally integrable, there exists $t > \beta$ close enough to $\beta$ such that $|g|_{Lip, C'} \int_0^{t-\beta} k(s) \, ds < 1/2$. Consequently, we get $
\delta(t)
\leq \delta(t) /2$, contradicting the definition of $t$. Thus $M = \beta$, concluding the proof when $g$ is bounded. 

If $g$ is not bounded, we can still consider $y_1$ and $y_2$ two locally bounded solutions of $(E_{k,f,g})$. Since $[0,\beta]$ is compact, they are bounded so that we can replace $g$ by $g \wedge M$ for $M$ large enough and we can conclude.
\end{proof}

Using Proposition~\ref{propo:uniqueness}, we can derive an approximate result for the solution of non-linear Volterra equation with non-continuous initial condition $f$ in terms of solution of non-linear Volterra equation with continuous initial condition. This result is stated in the following Lemma and is key to lifting the continuity assumption of Theorem \ref{thm:comparison:old}.

\begin{lemma}
\label{lemma:limitvolterra}
Suppose that $f$ is piecewise continuous and locally bounded and $g$ is bounded and locally Lipschitz continuous. There exists a decreasing sequence $f_n$ of continuous locally bounded functions such that $f_n(t) \to f(t)$ for any $t \geq 0$.

Moreover, if $x$ and $x_{n}$ are solutions of $(E_{k,f,g})$ and $(E_{k,f_n,g_n})$ respectively, then we have $x_n(t) \to x(t)$  for any $t \geq 0$.
\end{lemma}

The proof of Lemma~\ref{lemma:limitvolterra} is postponed until the proof of Theorem \ref{thm:comparison} below. We can now adapt Theorem \ref{thm:comparison:old} to our setup.

\begin{theorem}
\label{thm:comparison}
Suppose that $f_1$ and $f_2$ are locally bounded, and that $g_1$ and $g_2$ are positive locally Lipschitz continuous functions. Moreover, suppose that $f_1 \leq f_2$ and $g_1 \leq g_2$.

Suppose $y_1$ and $y_2$ are locally bounded solutions of $(E_{k,{f_1},{g_1}})$ and $(E_{k,{f_2},{g_2}})$ on some interval $[0,\beta]$, with $\beta > 0$. Then $y_1 \leq y_2$ on $[0,\beta]$.
\end{theorem}

\begin{proof}

Under the assumptions of Theorem \ref{thm:comparison}, we know by Proposition~\ref{propo:uniqueness} that the solutions $y_1$ and $y_2$ of $(E_{k,{f_1},{g_1}})$ and $(E_{k,{f_2},{g_2}})$ are unique. \\

\textbf{Step 1. } We first prove Theorem \ref{thm:comparison} under the additional assumptions that $f_1$ and $f_2$ are continuous and that $g_1$ and $g_2$ are continuous. In that case, we know that there exists $C>0$ such that $f_1 \geq - C$ and $f_2 \geq -C$ on $[0, \beta]$. Since $g_1$ and $g_2$ are positive, we further deduce that $y_1 \geq - C$ and $y_2 \geq -C$ on $[0, \beta]$. Let $F_i = f_i + C$, $G_i = g_i$ and $z_i = y_i + C$ for $i=1,2$. Then $z_i \geq 0$ and we can apply Theorem \ref{thm:comparison:old} to get $z_1 \leq z_2$ and we get the result. \\

\textbf{Step 2. Proof of Lemma~\ref{lemma:limitvolterra}} We are now ready to prove Lemma~\ref{lemma:limitvolterra} using the special case of Theorem \ref{thm:comparison} proved in Step 1. Applying Theorem \ref{thm:comparison} for continuous functions, we know that for any $t$, the sequence $(x_n(t))_n$ as defined in Lemma~\ref{lemma:limitvolterra} is decreasing. Moreover, it is bounded from below by $f(t)$. Therefore it converges towards a finite limit $x^*(t)$. 
Since $g$ is bounded and $k$ is locally integrable, $\int_0^t k(t-s) g(x_n(s)) \, ds \to \int_0^t k(t-s) g(x^*(s)) \, ds$. Therefore we get
\begin{align*}
x^*(t) = f(t) + \int_0^t k(t-s) g(x^*(s)) \, ds
\end{align*}
so that $x^*$ is solution of $(E_{k,f,g})$. Moreover, $f(t) \leq x^*(t) \leq x_1(t)$ so $x^*$ is locally bounded, which implies that $x^* = x$ by Proposition~\ref{propo:uniqueness}.\\

\textbf{Step 3.} Using Lemma~\ref{lemma:limitvolterra}, we see that Theorem \ref{thm:comparison} holds whenever $f_1$ and $f_2$ are locally bounded and $g_1$ and $g_2$ are bounded. Moreover, if $g_1$ and $g_2$ are not necessarily bounded, $y_1$ and $y_2$ are still locally bounded so we can replace $g_1$ and $g_2$ by $g_1 \wedge M$ and $g_2 \wedge M$ for $M$ large enough and we can conclude.
\end{proof}

We conclude this section by providing a comparison theorem used in the following proofs.
\begin{theorem}
\label{thm:comparison:2}
Suppose that $g$ is positive and locally Lipschitz continuous and that $f$ is locally bounded. We consider the equation $(E_{k,f,g})$ and we assume that it has a continuous solution $x$ on $[0,\beta]$ for some $\beta > 0$. Let $y_{+}$ and $y_{-}$ be two continuous functions on $[0,\beta]$ for some $\beta > 0$. We define
\begin{align*}
z_{\pm}(t) = f(t) + \int_0^t k(t-s) g(y_{\pm}(s)) \, ds.
\end{align*}
Then the following holds:
\begin{itemize}
\item If $z_+(t) \geq y_+(t)$ for all $t$, then $y_+(t) \leq x(t)$ for all $t$.
\item If $z_-(t) \leq y_-(t)$ for all $t$, then $y_-(t) \geq x(t)$ for all $t$.
\end{itemize}

\end{theorem}

\begin{proof}
Note that
\begin{align*}
y_\pm(t) = f_\pm(t) + \int_0^t k(t-s) g(y_{\pm}(s)) \, ds
\end{align*}
where $ f_\pm(t) = f(t) + y_\pm(t) - z_\pm(t)$, so that $y_\pm$ is a solution of $(E_{k,f_\pm,g})$. By definition of $y_\pm$ and $z_\pm$, we know that $f_\pm$ is locally bounded. Moreover, we have $f_+ \leq f$ (resp. $f \leq f_-$) whenever $z_+ \geq y_+$ (resp.  $z_- \leq y_-$). Therefore, we can apply Theorem \ref{thm:comparison}, which implies the result.
\end{proof}

\subsection{Proof of Theorem \ref{thm:concave_mi}}
\label{sec:proof:concave_mi}

We first prove that the market impact is increasing in $\gamma$. Recall that it is given by
\begin{align*}
MI(\gamma, t) = \int_0^t (1 + \lambda^{-1} (t-s)^{-\alpha}) (\gamma f(s) - r^*(\gamma, s)) \, ds.
\end{align*}
Therefore, it is enough to prove that $\gamma f(t) - r^*(\gamma, t)$ is increasing. Let $\gamma_1 < \gamma_2$. We plan to use Theorem \ref{thm:comparison:2} on the functions $u(\gamma_1, \cdot)$ and $u(\gamma_2, \cdot)$ as defined in Section \ref{sec:macroscopic:notation}. First, we define
\begin{align*}
\begin{cases}
y(t) = \MINL^{-1} \Big(\gamma_2 f(t) - \gamma_1 f(t) + \MINL(u(\gamma_1, t))\Big),
\\
z(t) =
\gamma_2 
F_{\alpha, \lambda}(t) + 
\int_0^t 
K_\alpha(t-s) 
G_{\alpha, \lambda}(y(s)) \, ds.
\end{cases}
\end{align*}
Thus we have
\begin{align*}
z(t) 
&= 
 \int_0^t \lambda^{-1} (t-s)^{-\alpha}(\gamma_2 f(s) - \MINL(y(s)) \, ds
= 
 \int_0^t \lambda^{-1} (t-s)^{-\alpha}(\gamma_1 f(s) - \MINL(u(\gamma_1, t)) \, ds
= 
u(\gamma_1, t)
\end{align*}
by definition of $u(\gamma_1, \cdot)$. But since $\MINL$ (and so $\MINL^{-1}$) is increasing and $(\gamma_2 - \gamma_1)f(t) \geq 0$, we get
$y(t) \geq u(\gamma_1, t) = z(t)$. Therefore, we can apply Theorem \ref{thm:comparison:2}, from which we obtain that $y(t) \geq u(\gamma_2, t)$. From the definition of $y$ and using again that $\MINL$ is increasing, we deduce that $r^*(\gamma_2, t) \leq \gamma_2 f(t) - \gamma_1 f(t) + r^*(\gamma_1, t)$, which concludes the proof.\\

We now prove that the market impact is concave in $\gamma$. Since $MI(\gamma, t) = \int_0^t (1 + \lambda^{-1} (t-s)^{-\alpha}) (\gamma f(s) - r^*(\gamma, s)) \, ds$, it is enough to prove that $\gamma f(t) - r^*(\gamma, t)$ is concave in $\gamma$, \textit{i.e.} that $r^*(\gamma, t)$ is convex in $\gamma$. Let $\gamma_1, \gamma_2$ and $0 < w < 1$. We want to prove that 
\begin{align*}
r^*(w \gamma_1 + (1-w)\gamma_2 , t) 
\leq 
w r^*(\gamma_1 , t) + (1-w) r^*(\gamma_2 , t).
\end{align*}
Since both sides are positive, and since $\MINL$ is non-decreasing and continuous on $[0, \infty)$, it is enough to prove that 
\begin{align*}
u(w \gamma_1 + (1-w)\gamma_2 , t) 
\leq 
\MINL^{-1}\Big(w \MINL(u(\gamma_1 , t)) + (1-w) \MINL(u(\gamma_2 , t))\Big).
\end{align*}
Let
\begin{align*}
\begin{cases}
y(t) = \MINL^{-1}\big(w \MINL(u(\gamma_1 , t)) + (1-w) \MINL(u(\gamma_2 , t))\big)
\\
z(t) = (w \gamma_1 + (1-w)\gamma_2)
F_{\alpha, \lambda}(t) + 
\int_0^t 
K_\alpha(t-s) 
G_{\alpha, \lambda}(y(s)) \, ds.
\end{cases}
\end{align*}
Since $\MINL$ is convex, we know that $\MINL^{-1}$ is concave and therefore
\begin{align*}
y(t) \geq w u(\gamma_1 , t) + (1-w) u(\gamma_2 , t).
\end{align*}
Moreover we have
\begin{align*}
G_{\alpha, \lambda} (y(s))
&=
-
\frac{\Gamma(1-\alpha)}{\lambda}  \MINL\Big(\MINL^{-1}\Big(w \MINL(u(\gamma_1 , t)) + (1-w) \MINL(u(\gamma_2 , t))\Big)\Big)
\\
&= 
-
\frac{\Gamma(1-\alpha)}{\lambda}  \Big(w \MINL(u(\gamma_1 , t)) + (1-w) \MINL(u(\gamma_2 , t))\Big)
\\
&= w G_{\alpha, \lambda} (u(\gamma_1 , t))
+
(1-w) G_{\alpha, \lambda} (u(\gamma_2 , t)).
\end{align*}
Therefore,
\begin{align*}
z(t) &= w \gamma_1 
F_{\alpha, \lambda}(t) + 
\int_0^t 
K_\alpha(t-s) 
G_{\alpha, \lambda}(u(\gamma_2 , s)) \, ds
\\
&\;\;\;\;\;\;\;\;\;\;\;\;\;\;\;\;
+ (1-w)\gamma_2
F_{\alpha, \lambda}(t) + 
\int_0^t 
K_\alpha(t-s) 
G_{\alpha, \lambda}(u(\gamma_2 , s)) \, ds.
\\
&= w u(\gamma_1 , t) + (1-w) u(\gamma_2 , t)
\end{align*}
so that $z(t) \leq y(t)$. Using Theorem \ref{thm:comparison:2}, we deduce that $y(t) \geq u(w \gamma_1 + (1-w)\gamma_2 , t)$, which concludes the proof.

\subsection{Proof of Theorem \ref{thm:sqrt_mi}}
\label{sec:proof:MI_sqrt}

First, note that since we have $f(t) = 1$ for $0 \leq t \leq 1$ and since we only consider $0 \leq t \leq 1$ in Theorem \ref{thm:sqrt_mi}, we can assume without loss of generalities that $f(t) = 1$ for any $t \geq 0$. Then we define, for $t \geq 0$
\begin{align*}
\widetilde{r}^*(\gamma, t) = \gamma^{-1} r^*(\gamma, \gamma^{\kappa} t)
\end{align*}
where $\upsilon = (1-\beta)\beta^{-1}(1-\alpha)^{-1} < 0$  is such that $((1-\alpha)\upsilon + 1 )\beta - 1 = 0$. For any $t \geq 0$,  we have
\begin{align*}
\widetilde{r}^*(\gamma, t)
&= \gamma^{-1} \MINL \Big( \int_0^{\gamma^{\upsilon} t} \kappa \lambda^{-1} (\gamma^{\upsilon} t-s)^{-\alpha}(\gamma - r^*(\gamma, s)) \, ds \Big)
\\
&= \gamma^{-1} \MINL \Big( \int_0^{t} \kappa \lambda^{-1} (\gamma^{\upsilon} t- \gamma^{\upsilon} s)^{-\alpha}(\gamma - r^*(\gamma, \gamma^{\upsilon} s)) \, \gamma^{\upsilon} ds \Big)
\\
&=  \gamma^{((1-\alpha)\upsilon + 1 )\beta - 1} \MINL \Big( \int_0^{t} \kappa \lambda^{-1} ( t -  s)^{-\alpha}(1 - \widetilde{r}^*(\gamma, s) \,  ds \Big).
\\
&=  \MINL \Big( \int_0^{t} \kappa \lambda^{-1} ( t -  s)^{-\alpha}(1 - \widetilde{r}^*(\gamma, s) \,  ds \Big).
\end{align*}
Therefore, for any $\gamma$, $\widetilde{r}^*(\gamma, \cdot)$ satisfies the same integro differential equation as $r^*(\gamma, \cdot)$. By Theorem \ref{thm:market_impact}, the solution is unique so we deduce that $\widetilde{r}^*(\gamma, t) = r^*(1, t)$ for any $\gamma$ and any $t$. Therefore,
\begin{align*}
r^*(\gamma, t) = \gamma r^*(1, \gamma^{-\upsilon} t).
\end{align*}
Since $\upsilon < 0$, it turns out that the behaviour of $r^*(\gamma, t)$ as $\gamma \to \infty$ is governed by the behaviour of $r^*(1, t)$ as $t \to \infty$. We first study the behaviour of $u$, which is defined in Section \ref{sec:macroscopic:notation}.

\begin{lemma}
\label{lem:asymptotic}
Suppose that $\MINL$ is continuous, convex and non-decreasing. Suppose also that $\MINL$ is differentiable on $(0, \infty)$. Let $\gamma > 0$ and $r^*$ be the (unique) solution of
\begin{align*}
r^*(t) = \MINL \Big ( \int_0^\infty \lambda^{-1} (t-s)^{-\alpha} (\gamma - r^*(s)) \, ds \Big).
\end{align*}
Then we have
\begin{align*}
\gamma
- \frac{\lambda \MINL^{-1}(\gamma) }{\Gamma(1-\alpha) \Gamma(\alpha) } \frac{1}{t^{1-\alpha}} 
  + O(t^{-2+\alpha})
\leq 
r^*(t) 
\leq 
\gamma
- \frac{\lambda \MINL^{-1}(\gamma) }{\Gamma(1-\alpha) \Gamma(\alpha) } \frac{1}{t^{1-\alpha}} 
  + O(t^{-2+\alpha}).
\end{align*}
\end{lemma}
Note that the assumptions of Lemma~\ref{lem:asymptotic} are slightly more general than the assumptions of Theorem \ref{thm:sqrt_mi}. The proof of Lemma~\ref{lem:asymptotic} is delayed until Section \ref{sec:proof:asymptotic}. The result is similar to previous studies \cite{handelsman1972asymptotic, olmstead1976asymptotic, olmstead1977nonlinear, soni1981asymptotic} that formally derived some asymptotics for several non-linear Volterra equations and the proof is inspired by \cite{gripenberg1981asymptotic} where similar expansions are rigorously proved in a different setup.  \\

We can apply the results from the previous section with $\gamma = 1$, and therefore $\MINL^{-1}(\gamma) = c^{-1/\beta}$, and with $\kappa\lambda^{-1}$ in place of $\lambda^{-1}$. Thus for any $0 < \varepsilon < 1$, there exists $t_\varepsilon > 0$ such that for all $t \geq t_\varepsilon$, we have
\begin{align}
\label{eq:bounds:r*}
1
- \frac{\lambda c^{-1/\beta} (1-\varepsilon)}{\kappa \Gamma(1-\alpha) \Gamma(\alpha)} \frac{1}{t^{1-\alpha}}
\leq r^*(1, t)
\leq 
1
- \frac{\lambda c^{-1/\beta} (1+\varepsilon)}{\kappa \Gamma(1-\alpha) \Gamma(\alpha)} \frac{1}{t^{1-\alpha}} .
\end{align}
Moreover, we have
\begin{align*}
MI(\gamma, t) 
&= \kappa \int_0^t (1 + \lambda^{-1} (t-s)^{-\alpha}) (\gamma - r^*(\gamma, s)) \, ds
\\
&= 
\kappa \int_0^t (\gamma - r^*(\gamma, s)) \, ds
+
c^{-1/\beta} [r^*(\gamma, t)]^{1/\beta}
\\
&= 
\kappa \gamma^{\upsilon+1}
\int_0^{\gamma^{-\upsilon} t} (1 - r^*(1, s)) \, ds
+
\gamma^{1/\beta} c^{-1/\beta}
[r^*(1,\gamma^{-\upsilon} t)]^{1/\beta}.
\end{align*}
We now fix $t > 0$ and we take the asymptotic $\gamma \to \infty$. Since $\upsilon < 0$, we have $\gamma^{-\upsilon} t \to \infty$ as $\gamma \to \infty$. Therefore, $\gamma^{-\upsilon} t \geq t_\varepsilon$ for all $\gamma \geq \gamma_\varepsilon$ for some fixed $\gamma_\varepsilon$. For such $\gamma$, we obtain by plugging \eqref{eq:bounds:r*}
\begin{align*}
MI(\gamma, t) 
&= 
\kappa 
\gamma^{\upsilon+1}
\int_0^{t_\varepsilon} (1 - r^*(1, s)) \, ds
+
\kappa 
\gamma^{\upsilon+1}
\int_{t_\varepsilon}^{\gamma^{-\upsilon} t} (1 - r^*(1, s)) \, ds
+
\gamma^{1/\beta}
[r^*(1,\gamma^{-\upsilon} t)]^{1/\beta}
\\
&\leq 
\kappa 
\gamma^{\upsilon+1}
\int_0^{t_\varepsilon} (1 - r^*(1, s)) \, ds
+
\kappa 
\gamma^{\upsilon+1}
\int_{t_\varepsilon}^{\gamma^{-\upsilon} t} \Big(\frac{\lambda c^{-1/\beta} (1+\varepsilon)}{\kappa \Gamma(1-\alpha) \Gamma(\alpha)} \frac{1}{s^{1-\alpha}} \Big) \, ds
\\
&\;\;\;\;\;\;\;\;\;\;\;\;\;\;\;\;\;\;\;\;+
\gamma^{1/\beta} c^{-1/\beta}
\Big[1
- \frac{\lambda c^{-1/\beta} (1-\varepsilon)}{\kappa \Gamma(1-\alpha) \Gamma(\alpha)} \frac{\gamma^{\upsilon(1-\alpha)}}{t^{1-\alpha}}
\Big]^{1/\beta}.
\end{align*}
The first integral is clearly dominated by $\gamma^{\upsilon + 1}$. For the second integral, explicit computations yield
\begin{align*}
\kappa \gamma^{\upsilon+1}
\int_{t_\varepsilon}^{\gamma^{-\upsilon} t} \Big(\frac{\lambda c^{-1/\beta} (1+\varepsilon)}{\kappa \Gamma(1-\alpha) \Gamma(\alpha)} \frac{1}{s^{1-\alpha}} \Big) \, ds
&=
\gamma^{\upsilon+1}
 \frac{\lambda c^{-1/\beta} (1+\varepsilon)}{\alpha\Gamma(1-\alpha) \Gamma(\alpha)}
\Big(
\gamma^{-\alpha\upsilon} t^{\alpha} - t_\varepsilon^{\alpha}
\Big)
\\
&=
\gamma^{1/\beta}
 \frac{\lambda c^{-1/\beta} (1+\varepsilon)}{\alpha\Gamma(1-\alpha) \Gamma(\alpha)}
 t^{\alpha} 
 + 
 O (\gamma^{\upsilon+1}),
\end{align*}
where we use that $(1-\alpha)\upsilon+1 = \beta^{-1}$ by definition of $\upsilon$. For the last term, we have
\begin{align*}
\gamma^{1/\beta} c^{-1/\beta}
\Big[1
- \frac{\lambda c^{-1/\beta}(1-\varepsilon)}{\kappa \Gamma(1-\alpha) \Gamma(\alpha)} \frac{\gamma^{\upsilon(1-\alpha)}}{t^{1-\alpha}}
\Big]^{1/\beta}
&=
\gamma^{1/\beta} c^{-1/\beta}
\Big[1
- 
\frac{1}{\beta}
\frac{\lambda c^{-1/\beta}(1-\varepsilon)}{\kappa \Gamma(1-\alpha) \Gamma(\alpha)} \frac{\gamma^{\upsilon(1-\alpha)}}{t^{1-\alpha}}
(1 + o(1))
\Big]
\\
&=
\gamma^{1/\beta} c^{-1/\beta} + O(\gamma^{\upsilon(1-\alpha) + 1/\beta}).
\end{align*}
Combining all this, we obtain
\begin{align*}
MI(\gamma, t) 
&\leq
\gamma^{1/\beta} c^{-1/\beta}
\Big(
1
+
\frac{\lambda  (1+\varepsilon)}{\alpha \Gamma(1-\alpha) \Gamma(\alpha)} t^{\alpha} 
\Big)
+
O (\gamma^{\upsilon+1} + \gamma^{\upsilon(1-\alpha)+1/\beta}).
\end{align*}
Note in addition that $\upsilon+1 < 1/\beta$ and $\upsilon(1-\alpha)+1/\beta < 1 / \beta$ so that 
\begin{align*}
\limsup\limits_{\gamma \to \infty}
c^{1/\beta} \gamma^{-1/\beta} MI(\gamma, t) 
\leq 
1
+
\frac{\lambda}{\alpha \Gamma(1-\alpha) \Gamma(\alpha)} t^{\alpha}.
\end{align*}
Similarly, we show that
\begin{align*}
\liminf\limits_{\gamma \to \infty}
c^{1/\beta} \gamma^{-1/\beta} MI(\gamma, t) 
\geq
1
+
\frac{\lambda}{\alpha \Gamma(1-\alpha) \Gamma(\alpha)} t^{\alpha},
\end{align*}
which concludes the proof.

\subsection{Proof of Lemma~\ref{lem:asymptotic}}
\label{sec:proof:asymptotic}

We suppose that $\MINL$ is continuous, convex and non-decreasing. For simplicity, we write $g(x) = \lambda^{-1}\Gamma(1-\alpha)( \gamma - \MINL(x))$, $K(t) = K_{\alpha}(t)$ and $\delta = \MINL^{-1}(\gamma)$. By analogy with the notation of Section \ref{sec:macroscopic:notation}, we introduce $u(t) = \MINL^{-1}(r^*(t))$ for all $t$ so that
\begin{align*}
u(t) = \int_0^\infty \lambda^{-1} (t-s)^{-\alpha} (\gamma - \MINL(u(s))) \, ds = \int_0^\infty K(s) g(u(s)) \, ds.
\end{align*} 

We first show that $u$ is non-decreasing and that for all $t$, $0 \leq u(t) \leq \delta$. We plan to use the following Lemma, adapted to our setup from the first lemma in \cite{gripenberg1981asymptotic}.
\begin{lemma}[\cite{gripenberg1981asymptotic}]
Let $g$ be a continuous nondecreasing function on $[0, A]$ for some $A > 0$ such that $g(A) = 0$ Then, there exists a unique, continuous solution $u$ of $(E_{K_\alpha,0,g})$. Moreover, $u$ is nondecreasing and satisfies $0 \leq u(t) \leq A$.
\end{lemma}

Since $g$ is decreasing and $g(\delta) = 0$, we can use this lemma and we get a nondecreasing continuous solution $\widehat{u}$ of $(E_{K_\alpha,0,g})$ such that $0 \leq \widehat{u} \leq \delta$. But since $g$ is continuous, $(E_{K_\alpha,0,g})$ has only one solution so $u=\widehat{u}$. Moreover, since $u$ is nondecreasing and bounded, there exists $u_\infty \leq \delta$ such that $u(t) \to u_\infty$ as $t\to \infty$. Note also that $g$ is decreasing so $G(u(s)) \geq G(u_\infty)$ for all $s \geq 0$ and thus
\begin{align*}
u(t) = \int_0^t K(t-s) g(u(s))\, ds \geq g(u_\infty) \int_0^t K(s)\, ds.
\end{align*}
If $G(u_\infty) > 0$, we have $u(t) \to \infty$ by definition of $K$. Since $u$ is bounded, we deduce that necessarily $G(u_\infty) \leq 0$ and therefore $u_\infty = \delta$.\\

We now focus on the asymptotic expansion of $v(t) = \delta - u(t)$. Let $x$ be the solution of $x(t) = \delta - \frac{\Gamma(1-\alpha) \MINL'(\delta)}{\lambda} \int_0^t K(t-s) x(s)\, ds$. By definition of $K$, it is given explicitly by
\begin{align*}
x(t) = \delta E_{1-\alpha, 1}
\Big( -\frac{\Gamma(1-\alpha)\MINL'(\delta)}{\lambda} t^{1-\alpha} \Big)
\end{align*}
where $E_{1-\alpha, \beta}$ is the $(1-\alpha, \beta)$ Mittag-Leffler function, whenever $\beta \geq 0$, which is defined by
\begin{align}
\label{eq:asymptotic:mittag}
E_{1-\alpha, \beta}(z) = \sum_{k=0}^\infty \frac{z^k}{\Gamma(\alpha k + \beta)}.
\end{align}
We refer to \cite{haubold2011mittag} for a complete presentation of the Mittag-Leffler function. From Equation (6.11) in \cite{haubold2011mittag}, we have for any $N > 1$ the following asymptotic expansion as $z \to \infty$
\begin{align*}
E_{1-\alpha, \beta}(z) = - \sum_{k=1}^{N}
\frac{1}{z^k \Gamma(\beta-k(1-\alpha))} + O(z^{-N-1}).
\end{align*} 
In particular, we obtain
\begin{align*}
x(t) 
 &=
 \frac{\lambda\delta}{\Gamma(1-\alpha) \Gamma(\alpha)\MINL'(\delta)} \frac{1}{t^{1-\alpha}} 
- \frac{\lambda^2\delta}{\Gamma(1-\alpha)^2\MINL'(\delta)^2 \Gamma(2\alpha-1)} \frac{1}{t^{2-2\alpha}} 
+
O\Big(
\frac{1}{{t^{3-3\alpha}}}
\Big).
\end{align*}
Note that since $\MINL$ is convex, $g$ is concave and therefore
$g(\delta - x) \leq \lambda^{-1} {\Gamma(1-\alpha) \MINL'(\delta)} x$ for all $x \geq 0$. Applying this to $v$ ensures that 
\begin{align*}
v(t) &= \delta - \int_0^t K(t-s) g(\delta - v(s))\, ds
\geq \delta - \frac{\Gamma(1-\alpha) \MINL'(\delta)}{\lambda} \int_0^t K(t-s) v(s)\, ds.
\end{align*}
By definition of $x$ and using Theorem \ref{thm:comparison:2}, we deduce that 
$v(t) \geq x(t)
$ and therefore $u(t) \leq \delta - x(t)$. Since $r^*(t) = \MINL(u(t))$ and since $\MINL$ is increasing and twice differentiable on $(0,\infty)$, we get
\begin{align}
\label{eq:inequalityr*:sup}
r^*(t) &\leq \MINL( \delta - x(t))
= \gamma - \MINL'(\delta) x(t) + O(x(t)^2)
=
\gamma
- \frac{\lambda\delta}{\Gamma(1-\alpha) \Gamma(\alpha)} \frac{1}{t^{1-\alpha}} 
+
O\Big(
\frac{1}{{t^{2-2\alpha}}}
\Big).
\end{align}

We now prove that Inequality \ref{eq:inequalityr*:sup} can be reversed. Let  $\MINL_-' < \MINL'(\delta)$ be fixed. Then there exists $\varepsilon > 0$ such that for all $0 \leq x \leq \varepsilon$, we have
$g(\delta - x) \geq \lambda^{-1} \Gamma(1-\alpha) \MINL_-' x
$. Since $v(t) \to 0$ as $t \to \infty$, there also exists $t_\varepsilon > 0$ such that $v(t) \leq \varepsilon$ whenever $t \geq t_\varepsilon$. Therefore, we get
\begin{align*}
v(t) &= \delta - \int_0^t K(t-s) G(\delta - v(s))\, ds
\\
& \leq \delta - \frac{\Gamma(1-\alpha) \MINL_-'}{\lambda} \int_0^t K(t-s) v(s)\, ds
+ \int_0^{t_\varepsilon} K(t-s) \Big(\frac{\Gamma(1-\alpha) \MINL_-'}{\lambda} v(s) - G(\delta - v(s)) \Big)\, ds.
\end{align*}
Using again Theorem \ref{thm:comparison:2}, we deduce that 
$v(t) \leq y(t)$ where $y$ is the solution of 
\begin{align}
\label{eq:def:yasymptotic}
y(t) = 
\delta + h(t) - \frac{\Gamma(1-\alpha) \MINL_-'}{\lambda} \int_0^t K(t-s) y(s)\, ds
\end{align}
where $h(t) = \int_0^{t_\varepsilon} K(t-s) \big(\lambda^{-1}{\Gamma(1-\alpha) \MINL_-'} v(s) - G(\delta - v(s)) \big)\, ds$ which we rewrite as $h(t) = \int_0^{t} K(t-s) w(s)\, ds$ with $w(t) = \big(\lambda^{-1}\Gamma(1-\alpha) \MINL_-' v(t) - G(\delta - v(t)) \big) \1_{t \leq t_\varepsilon}$. 
Remark that Equation \eqref{eq:def:yasymptotic} is linear in $y$. Therefore, we decompose $y(t) = y_1(t) + y_2(t)$ where $y_1$ and $y_2$ satisfy
\begin{align*}
\begin{cases}
y_1(t) = \delta - \frac{\Gamma(1-\alpha) \MINL_-'}{\lambda} \int_0^t K(t-s) y_1(s)\, ds
\\
y_2(t) = h(t) - \frac{\Gamma(1-\alpha) \MINL_-'}{\lambda} \int_0^t K(t-s) y_2(s)\, ds
\end{cases}
\end{align*}
We can solve explicitly the first one and we get
\begin{align}
\label{eq:def:y1}
y_1(t) = \delta E_{1-\alpha, 1}
\Big( -\frac{\Gamma(1-\alpha)\MINL_-'}{\lambda} t^{1-\alpha} \Big).
\end{align}
For the second one, we introduce artificially
\begin{align*}
B(t) = \frac{\Gamma(1-\alpha) \MINL_-'}{\lambda t^{\alpha}} E_{1-\alpha, 1-\alpha} \Big( - \frac{\Gamma(1-\alpha) \MINL_-'}{\lambda } t^{1-\alpha} \Big)
\end{align*}
which is solution of the integro-differential equation
\begin{align*}
B(t) = \frac{\Gamma(1-\alpha) \MINL_-'}{\lambda} \Big[ K(t) - \int_0^t K(t-s) B(s)\, ds \Big].
\end{align*}
Plugging this into the definition of $y_2$ and using that $h(t) = \int_0^{t} K(t-s) w(s)\, ds$, we get by algebraic computations
\begin{align*}
y_2(t) 
&= f(t) - \int_0^t B(t-s) f(s)\, ds
\\
&=
\int_0^{t} \Big( K(t-s) - \int_0^{t-s} K(t-s-u) B(u)\, du\Big)\, w(s)\, ds
\\
&=
\frac{\lambda} {\Gamma(1-\alpha) \MINL_-'}
\int_0^{t} B(t-s) w(s)\, ds.
\end{align*}
Recall that $w$ is bounded and vanishes outside of $[0, t_\varepsilon]$. Therefore, using that $B$ is decreasing on $(0,\infty)$, we deduce that 
\begin{align}
\label{eq:bound:y2}
0 \leq y_2(t) 
&\leq 
 |\!|w|\!|_{\infty}
\frac{\lambda t_\varepsilon |\!|w|\!|_{\infty}} {\Gamma(1-\alpha) \MINL_-'}
 B_2(t-t_\varepsilon).
\end{align}
Plugging the expansion \eqref{eq:asymptotic:mittag} into \eqref{eq:def:y1} and \eqref{eq:bound:y2}, we obtain
\begin{align*}
\begin{cases}
y_1(t) =
 \frac{\lambda\delta}{\Gamma(1-\alpha) \Gamma(\alpha)\MINL'_-} \frac{1}{t^{1-\alpha}} 
- \frac{\lambda^2\delta}{\Gamma(1-\alpha)^2\MINL'_-{\!}^2 \Gamma(2\alpha-1)} \frac{1}{t^{2-2\alpha}} 
+
O\Big(
\frac{1}{{t^{3-3\alpha}}}
\Big),
\\
y_2(t) = O(t^{-2+\alpha}).
\end{cases}
\end{align*}
Eventually, since $r^*(t) = \MINL(u(t))$ and since $\MINL$ is convex, we get
\begin{align*}
r^*(t) &\geq \MINL( \delta - y(t))
\geq = \gamma - \MINL'(\delta) y(t)
= \gamma - \MINL'(\delta) y_1(t) - \MINL'(\delta) y_2(t).
\end{align*}
Using the asymptotic expansion of $y_1$ and $y_2$, we get
\begin{align*}
r^*(t) &\geq 
\gamma
- \frac{\lambda\delta \MINL'(\delta)}{\Gamma(1-\alpha) \Gamma(\alpha) \MINL_-'} \frac{1}{t^{1-\alpha}} 
  + O(t^{-2+\alpha}).
\end{align*}

\subsection{Proof of Theorem \ref{thm:small:gamma}}

Recall from \eqref{eq:def:u} that $
u(\gamma, t) = \int_0^t \lambda^{-1} (t-s)^{-\alpha}(\gamma f(s) - \MINL(u(\gamma, t))) \, ds$ 
and $r^*(\gamma, t) = \MINL(u(\gamma, t))$. As $\MINL$ is positive, $G_{\alpha, \lambda}$ is negative and therefore we obtain from \eqref{eq:equation:u} that
\begin{align*}
u(\gamma, t) \leq
\gamma
F_{\alpha, \lambda}(t).
\end{align*}
Since $\MINL$ is increasing, we also get
\begin{align*}
u(\gamma, t)  
&\geq
\int_0^t \lambda^{-1} (t-s)^{-\alpha}(\gamma f(s) - \MINL(
 \gamma F_{\alpha, \lambda}(s) )) \, ds
=
 \gamma F_{\alpha, \lambda}(t) - \int_0^t \lambda^{-1} (t-s)^{-\alpha} \MINL(
 \gamma F_{\alpha, \lambda}(s) ) \, ds.
\end{align*}
Therefore, we get
\begin{align*}
\gamma F_{\alpha, \lambda}(t) - \int_0^t \lambda^{-1} (t-s)^{-\alpha} \MINL(
 \gamma F_{\alpha, \lambda}(s) ) \, ds
\leq
u(\gamma, t)  
\leq
\gamma F_{\alpha, \lambda}(t)
\end{align*}
which yields
\begin{align*}
\MINL\Big(\gamma F_{\alpha, \lambda}(t) - \int_0^t \lambda^{-1} (t-s)^{-\alpha} \MINL(
 \gamma F_{\alpha, \lambda}(s) ) \, ds\Big)
\leq
r^*(\gamma, t)  
\leq
\MINL\Big(\gamma F_{\alpha, \lambda}(t)\Big).
\end{align*}
We conclude by plugging these inequalities into $
MI(\gamma, t) 
= 
\int_0^t 
(1 + \lambda^{-1} (t-s)^{-\alpha}) (\gamma f(s) - r(\gamma, s))
\, ds
$.

\end{document}